\documentclass[conference,compsoc]{IEEEtran}

\usepackage[english]{babel}
\usepackage{blindtext}
\usepackage{graphicx}
\graphicspath{ {Figures/} }
\usepackage{caption}
\usepackage{subcaption}
\usepackage{stackengine}
\usepackage{float}
\usepackage{amssymb}
\usepackage{hyperref}

\usepackage[hang,flushmargin]{footmisc}
\usepackage{color}
\usepackage{microtype,pdfsync}
  \usepackage{amsthm,bm}
  \usepackage{fullpage}
\usepackage{tabu} % detect if table is in math mode
\usepackage{graphicx}
\usepackage{mathtools}
\usepackage{footnote}
\usepackage{framed}
\usepackage{xspace}
\usepackage{multirow}
\usepackage{enumitem}

\usepackage{tikz}
\usetikzlibrary{calc,positioning,shapes,shadows,arrows,fit}

\usepackage{adjustbox}
\usepackage{amsmath}
\usepackage{amsthm}
\usepackage{anyfontsize}
\usepackage{booktabs}
\usepackage{changepage}   % adjustwidth
\usepackage{color}
\usepackage{colortbl}
\usepackage{framed}
\usepackage{letltxmacro}
\usepackage{mathtools}
\usepackage{microtype}
\usepackage{scalefnt}
\usepackage{setspace}
\usepackage{xcolor}
\usepackage{xspace}
\usepackage{wrapfig}
\usepackage{subcaption}
\usepackage[most]{tcolorbox}
\usepackage{lipsum}

\usepackage{algorithm}
\usepackage{algorithmic}

\newcommand{\ignore}[1]{}

\usepackage[font={small}]{caption}
\numberwithin{equation}{section}

\newcommand{\Z}{\mathbb{Z}}
 
\renewcommand{\paragraph}[1]{\medskip\noindent {\bf {#1}}}

\newcommand{\calA}{\ensuremath{\mathcal{A}}}
\newcommand{\calB}{\ensuremath{\mathcal{B}}}
\newcommand{\calC}{\ensuremath{\mathcal{C}}}

\newcommand{\zo}{\ensuremath{\{0,1\}}} % bits

  \theoremstyle{plain} \newtheorem{theorem}{Theorem}[section] 
  \newtheorem{lemma}[theorem]{Lemma}

  \theoremstyle{definition}  
  
  \newtheorem{definition}[theorem]{Definition}  

    \newcommand{\Theorem}[1]{\hyperref[#1]{Theorem~\ref*{#1}}}
    \newcommand{\Lemma}[1]{\hyperref[#1]{Lemma~\ref*{#1}}}
    \newcommand{\Corollary}[1]{\hyperref[#1]{Corollary~\ref*{#1}}}
    \newcommand{\Definition}[1]{\hyperref[#1]{Definition~\ref*{#1}}}
    \newcommand{\Example}[1]{\hyperref[#1]{Example~\ref*{#1}}}
    \newcommand{\Remark}[1]{\hyperref[#1]{Remark~\ref*{#1}}}
    \newcommand{\Fact}[1]{\hyperref[#1]{Fact~\ref*{#1}}}
    \newcommand{\Table}[1]{\hyperref[#1]{Table~\ref*{#1}}}
    \newcommand{\Figure}[1]{\hyperref[#1]{Fig.~\ref*{#1}}}
    \newcommand{\Section}[1]{\hyperref[#1]{Section~\ref*{#1}}}
    \newcommand{\Sections}[1]{\hyperref[#1]{Sections~\ref*{#1}}}
    \newcommand{\Appendix}[1]{\hyperref[#1]{Appendix~\ref*{#1}}}
    \newcommand{\Protocol}[1]{\hyperref[#1]{Protocol~\ref*{#1}}}
    \newcommand{\Equation}[1]{\hyperref[#1]{{(\ref*{#1})}}}

\newcommand{\esm}[1]{\ensuremath{#1}}
\newcommand{\ms}[1]{\esm{\mathsf{#1}}}

\newcommand{\tab}{\hspace*{4mm}}

\newcommand{\etal}{et~al.\xspace}

\definecolor{todocolor}{rgb}{0.8,0,0}
\definecolor{darkgreen}{rgb}{0.09, 0.45, 0.27}
\definecolor{neilcolor}{rgb}{0.7, 0.2, 0.27}

\definecolor{donecolor}{rgb}{0,0,0.8}

\newcommand{\HIDE}[1]{}

\newlist{myitemize}{itemize}{3}
\newlist{myenumerate}{enumerate}{3}
\setlist[myitemize]{label=$\bullet$,topsep=1.0ex,itemsep=0.0ex,leftmargin=1em}
\setlist[myenumerate]{label=(\arabic*),itemsep=0.0em,topsep=0.5ex,leftmargin=1.6em}

\title{Strong Anonymity for Mesh Messaging}

\author{
    Neil Perry\\
    Stanford University
    \and
    Bruce Spang\\
    Stanford University
    \and
    Saba Eskandarian\\
    UNC Chapel Hill
    \and
    Dan Boneh\\
    Stanford University
}
\begin{document}

\maketitle

\begin{abstract}
Messaging systems built on mesh networks consisting of smartphones communicating over Bluetooth have been used by protesters around the world after governments have disrupted Internet connectivity. Unfortunately, existing systems have been shown to be insecure; most concerningly by not adequately hiding metadata. This is further complicated by the fact that wireless communication such as Bluetooth is inherently a broadcasting medium. In this paper, we present a new threat model that captures the security requirements of protesters in this setting. We then provide a solution that satisfies the required security properties, hides all relevant metadata, scales to moderately sized protests, and supports group messaging. This is achieved by broadcasting all messages in a way that limits the overhead of duplicate messages, ensuring that ciphertexts do not leak metadata, and limiting what can be learned by observing user behavior. We also build a model of our system and numerically evaluate it to support our claims and analyze how many users it supports. Finally, we discuss further extensions that remove potential bottlenecks in scaling and support substantially more users.
\end{abstract}

\section{Introduction}

What would a messaging system designed for the safety of protesters look like? Specifically, what would this look like in a setting where access to the Internet has been disabled? In order to disrupt the organization of protests, governments have often disabled Internet access. To stay in contact, protesters in Hong Kong, India, Iran, Nigeria, Thailand, and the US have tried to adopt messaging apps such as FireChat and Bridgefy that form a mesh network through their phones’ Bluetooth antennas as a way to organize and communicate \cite{wakefield_2019, dewast_2014, butcher_2020, goodin_2020, bhushan_2019}, although some have had problems successfully using them \cite{trouble_using}. These apps were originally developed for non-adversarial environments, like a music festival or a flight, allowing users to talk to friends when there is no Internet. Several of them claimed to support end-to-end encryption and to keep their users anonymous; however, attacks have been demonstrated against Bridgefy \cite{bridgefy, breakingbridgefy, eikenberg_albrecht_paterson_2021}. In addition to this, the use cases of protesters are different from those of typical users. Protesters may only have access to low-end devices and not have access to cellular or Wi-Fi networks, but still need to be able to reliably communicate. Protesting can also be dangerous and keeping protesters’ identities private and not leaving records for a government to trace is critical to their safety \cite{joseph_2017, guynn_2016}. Therefore, not only must all communications be kept private, but metadata (i.e. a set of data that describes and gives information about other data) about these communications must not be leaked in order to allow protesters to safely use the system. As Michael Hayden, former director of the CIA and NSA, once said “We kill people based on metadata” \cite{ferran_2014}. There is no existing security definition for this problem.

Without security definitions stating what an adversary is capable of and what guarantees need to be made to protesters using the system, it is impossible to design one that adequately meets their needs. In addition, there is little prior work to build on. Existing anonymous communication systems rely on many tools that are not available in this setting such as centralized servers, public key infrastructures, and devices with relatively high computation and networking resources. The users of this system face the strongest adversaries such as nation-state actors and have the largest privacy requirements.

Since these issues are not addressed by the repurposed apps, this leads to the question, can a secure messaging system be built over a mesh network that provides the needed anonymity for protesters in the face of a nation-state adversary? We propose a security model that meets the needs of protesters in this setting which captures and formalizes these goals. This is useful since these systems are already being used in the real world. We found a solution that meets all of the requirements in our security model. Our system supports 200 users sending messages every 30 seconds under very pessimistic assumptions on how tightly packed a crowd of protesters is and supports more than 3000 users sending messages every minute under more realistic assumptions. Under the right conditions, our system can support more than 5000 users. Messages are delivered within several seconds in small protests and within two minutes in very large protests. Our system accommodates protests of moderate size, providing reliable message delivery in a reasonable amount of time. For very large protests (several tens of thousands of people or more), it is likely that these crowds would not form a fully connected graph and can be thought of as many instances of these moderately sized protests, making our system a potential option for these environments as well.

This is accomplished by broadcasting messages in a non-traditional way to minimize the impact of duplicate messages in the network. Our system hides the content of all messages and relevant metadata such as the intended receiver, timing information, and if users are sending messages by generating cover traffic that completely prevents traffic correlation. We support these results by measuring the performance of our system in models that capture key components of scenarios of realistic protests and outline ways to achieve further scaling.

In summary, we make the following contributions:
\begin{myenumerate}
  \item Define security for an anonymous communication platform running on a mesh network (Section \ref{sec:design})
  \item Propose a solution that meets this security definition (Section \ref{sec:system-description})
  \item Prove the security of our solution (Appendix \ref{sec:security_games})
  \item Build a model of the system (Section \ref{sec:implementation})
  \item Analyze the model and show that it scales to moderately sized protests (Section \ref{sec:eval})
  \item Describe extensions for further scaling (Section \ref{sec:extensions})
 \end{myenumerate}
 
% All the code and data for this paper, including the proof of concept code, model, and measurements, are publicly available at \url{https://anonymous.4open.science/r/ms_and_prototype_and_measurements-6370/}.
All the code and data for this paper, including the proof of concept code, model, and measurements, are publicly available at \url{https://github.com/NeilAPerry/ms_and_prototype_and_measurements}
\section{Design Goals}
\label{sec:design}

%It is useful to define the underlying system that we are trying to make secure and anonymous before explaining these requirements. 
%commented out the preceding line since it doesn't seem like it was meant to be part of the text ~Saba

We define a \emph{mesh messaging} system to be a system which allows its users to exchange messages among each other by relying only on their own devices and not on an external cellular/Wi-Fi network. This includes sending messages to users within direct range of a phone's radios as well as others who may only be reachable through multiple hops among adjacent phones. In this section, we describe what additional requirements are needed for one of these systems to allow for anonymous communication against a nation-state adversary.

An adversary that controls some subset of users may try the following attacks:
\begin{myitemize}
    \item Learn with whom a particular user is communicating
    \item Learn information about the communication patterns of users not under the adversary’s control
    \item Impersonate users not under the adversary’s control
    \item Modify a message sent from one user to another
    % \item disrupt communication for a particular user or all users by launching a denial of service attack
\end{myitemize}

%looks like in the latest version the whole thing takes place offline, so commenting out the next paragraph ~Saba
%In this work, we will distinguish between aspects of a mesh messaging system: a \emph{preparation} phase and a \emph{mesh} phase. During the preparation phase, users can access the internet and use a centralized system to establish relationships with friends, exchange keys, etc. The mesh phase takes place after internet connections have been severed and users can only rely on technologies available on their devices, e.g., Bluetooth.

After outlining potential attacks, it is important to consider the security and performance requirements for a mesh messaging system in greater detail. Given an adversary that may attempt any of these attacks, a mesh messaging system that respects privacy must satisfy two broad classes of design goals:
\vspace{-1ex}
\begin{myitemize}
\item \emph{Security requirements.} A private mesh messaging system should reveal no information about a given user's message contents or messaging behavior metadata. 

\item \emph{Performance requirements.} An effective mesh messaging system must ensure timely delivery of messages. This means that messages must be delivered in time proportional to the diameter of the graph. This must hold even as the number of users grows and the network experiences transient disruptions.
\end{myitemize}
\vspace{-1ex}
These security requirements establish a threat model that can be used in any work on messaging over a mesh network when security and privacy are needed. 

\paragraph{The adversary's capabilities.} A mesh messaging protocol for protesters must be resilient to attacks by well-resourced nation-state adversaries. Thus, we consider an attacker who can observe all radio communication in the vicinity where mesh messaging is being used, e.g., the location of a protest. Moreover, we assume the adversary can deploy rogue users who can participate in the messaging protocol and can arbitrarily deviate from the specified protocol. In particular, adversary-controlled users can (i) collude with and communicate with each other, (ii) send messages to each other or other users, and (iii) alter, delay, or drop messages that are routed through them. We exclude Denial of Service attacks. Defending against low-level denial of service attacks, e.g., aimed at disabling cellular radios, is orthogonal to our threat model as we only consider security and performance properties with regard to users who are able to connect to the mesh network. 

\paragraph{Security requirements.} Defending against a nation state adversary requires that the following data not be leaked: 1) content of messages; 2) communication links between users; and, 3) information from message timing. Furthermore, it is imperative that tampering with message contents is detected. The total number of mesh members is not hidden as this will be readily discernible from the overall traffic volume.  Ultimately the system must leak no information about a user’s behavior, except that they connected to the mesh messaging system.

We formalize our security goals by requiring the following two security notions:
\vspace{-1ex}
\begin{myitemize}
    \item \emph{Mesh confidentiality}: an adversary should be unable to learn anything about who users communicate with and what messages they send.
    
    \item \emph{Mesh integrity}: any message successfully delivered in our system must not have been tampered with. Conversely, any message that is tampered with by an adversary must be detected by the recipient.
\end{myitemize}
\vspace{-1ex}

We give formal game-based security definitions for these two notions in Appendix~\ref{sec:security_games}.

\paragraph{Performance requirements.}
Our system must be able to scale to large protests. This means maximizing the \emph{Capacity} of the network. Capacity is defined to be the maximum number of devices that the protocol can support and is dependent on factors such as how close people are in the crowd, how frequently messages and cover traffic is sent, and how frequently message digests used in broadcasting are exchanged between neighbors. These will be analyzed in Section \ref{sec:eval}. This is not necessarily achieved by  minimizing total communication costs, i.e., it is preferred to have a communication scheme that has relatively high average bandwidth usage but stays below a failure threshold than a communication system that has low average bandwidth but frequently crosses the failure threshold. Therefore, we want to minimize some high percentile of bandwidth usage on a device such as the 95th percentile. This failure threshold is defined by the underlying communications infrastructure, i.e., Bluetooth or Wi-Fi Direct. Computation costs must also be kept low enough that they are not the bottleneck on a user's device, i.e. devices can decrypt more incoming messages than the network can support each second. 

Finally, we require successful and timely delivery of messages. That is, whenever the network topology of a mesh network allows for delivery of a message, the message should be delivered, barring malicious interference or the network being over capacity. We capture the notion of a message being deliverable by defining a \emph{good path} through the network, and use this to state a definition of correctness for mesh networks. 
 
\begin{definition}[Good Path]
A set of graphs $G_1=(V,E_1), ..., G_n = (V, E_n)$ on the same vertex set $V$ has a \emph{good path} from node $a\in V$ to node $b\in V$ if a path can be walked from $a$ to $b$ by only using edges in $E_i$ for the $i$th edge in the path. We allow the walk to ``pause'' or ``wait'' by taking no edge (equivalent to taking a self edge) in the $i$th round.
\end{definition}
 
\begin{definition}[Correctness, informal]
Let the sequence of graphs $G_1,...,G_n$ on the same vertex set $V$ represent the state of a mesh network over time. We say that a mesh messaging scheme has \emph{correctness} if for any users $A,B\in V$ and for any sequence of graphs $G_i$ on vertex set $V$, messages from $A$ to $B$ are eventually delivered if there is a good path from $A$ to $B$ in $G_1,...,G_n$ and all nodes follow the protocol.
\end{definition}

\section{Anonymous Mesh Messaging}
\label{sec:system-description}

In this section, we discuss how the protocol conceptually works. Then, we discuss specific aspects of the design such as key distribution, the main control loop of the protocol, and show how group messaging is supported. Finally, we show why this design meets our security requirements.

\subsection{Protocol Overview}
The protocol can be run over various underlying network mediums such as Bluetooth or Wi-Fi Direct. It has 3 separate phases: \emph{Key Distribution}, \emph{Messaging}, and \emph{Distancing}. These three phases take place before the protest, during the protest, and after the protest. They work as follows:
\vspace{-1ex}
\begin{myitemize}
	\item \emph{Key Distribution Phase}: Users exchange public keys and add each other to their friends lists. This process is done in two phases and is described in Section \ref{sec:key_distribution}.
	\item \emph{Messaging Phase}: Users can message other users in their friends list.
	\item \emph{Distancing Phase}: Users change their Bluetooth address and delete all data from the protest to unlink themselves from traffic that may have been recorded.
\end{myitemize}
\vspace{-1ex}
 Our primary focus is the Messaging phase of the protocol; we will describe it in detail in Section \ref{sec:active_phase}.

During the protest, every device broadcasts a fixed-size dummy message at a fixed interval. If a device has a real message to send, it encrypts the message in a way that hides the identity of the receiver. In order to hide the identity of the receiver, ensure message integrity, and ensure that users cannot impersonate one another, a secure signcryption with the anonymity property that nothing about public keys can be learned from ciphertexts must be used. Appendix \ref{sec:public_key_anonymity} further explains this anonymity property and a specific implementation of the cipher is given in Appendix \ref{sec:signcryption}. Then, the device waits until it is time to send a dummy message and sends this real message instead by broadcasting it to all of its neighbors. Devices send messages using a reliable broadcast algorithm described in Section \ref{broadcasting_algo}, which aims to eventually deliver every message to all other devices.
When a device receives a message, it checks if the message is for itself by attempting to decrypt it and processes the message if it is the intended receiver. It is important to do this in a way that does not leak information via timing side channels. The risk is that a device does more work when a decryption succeeds, which can be observed by the device taking longer to respond to messages, e.g., as recently demonstrated against the Zcash cryptocurrency~\cite{TBP20}.

\subsection{Broadcasting Algorithm} \label{broadcasting_algo}
Intuitively, information is revealed about a receiver whenever someone does not receive information. Traditional unicasting (i.e. one sender and one receiver), leaks exactly who is communicating with one another. The unicast routing algorithms that scale to large mesh networks also rely on the sender announcing that they want to talk to the receiver to build routes on demand \cite{venkat_mohan_kasiviswanath_2011} which breaks our security definition. We discuss unicasting alternatives and why they fail in Appendix \ref{unicast_attacks}. Therefore, we must broadcast messages.

First, we discuss a na\"ive approach to broadcasting that we call \emph{Simple Broadcasting} and its shortcomings before introducing a more sophisticated technique. A natural way to broadcast a message is to have a device send it to all of its neighbors and whenever it receives a message, forward it to all of its neighbors. Unfortunately, this results in an unreasonably high number of duplicate message transmissions which is quadratic in the number of messages sent. Given the low bandwidth of our network, another technique is needed.

Now, we describe the \emph{Smart Broadcasting} algorithm we use to deliver messages to all devices. This algorithm is based on the anti-entropy phase of \cite{BHO+1999}. In the algorithm, devices store a list of messages and a digest of these messages. The algorithm works by periodically having devices share a digest with their neighbors. We refer to the digest as the \ms{DS} in our protocol, representing a data structure that indicates which messages have been received or generated. The digest is a bitstring where the $i$th bit is 1 when a message hashes to the $i$th index of a hash table and 0 otherwise. If a device receives a digest that contains messages that it has not received, it asks the sender of the digest to send the missing messages. We refer to this as a \ms{request}. Smart Broadcasting scales with the number of nodes in the network: avoiding the quadratic bandwidth blowup of Simple Broadcasting.
 %Our system uses Smart broadcasting for its reliability properties and its increased capacities at many spacing intervals relative to Simple broadcasting, which can be seen in Figure \ref{fig:capacities}.

\subsection{Key Distribution} \label{sec:key_distribution}

Users in our system need a way of obtaining and verifying others' keys. In this section, we propose a two phase scheme for before and after Internet access has been disabled. The first phase uses an efficient Public Key Infrastructure (PKI) accessed through the Internet and the second phase uses a scalable mechanism for in-person key exchange after Internet access has been removed. Regardless of which phase was used to exchange keys, there is no way for an untrusted adversary to know which users' public keys were used during a protest. Implementations must defend against side channel attacks that may leak the recipient in order for this guarantee to hold. Therefore, these keys only need to be exchanged once among friends and can then be used at all subsequent protests.

In the first phase, a natural idea is to use a central trusted server like Keybase \cite{keybase}. This conveys the traditional advantages of users being able to quickly add many other users; but also has many downsides in this setting. Having a public directory of users (even if pseudonyms are used) provides a record that may link actual user identities to protests. This does not meet our setting's privacy requirements, requiring us to use a PKI such as CONIKS \cite{coniks}. CONIKS can hide the list of identities using the system while providing key verification for end-to-end encrypted communication systems. Since these solutions fail after Internet access is lost, we need an alternative mechanism to switch to when there is no Internet access.

Needing a solution that does not rely on any existing infrastructure, including the Internet and PKIs, we propose the following key distribution system. Users have no choice but to initially share keys in person, e.g., by scanning QR codes on each others' devices via their camera. However, once users have added others to their contacts, we accelerate the key distribution process by augmenting in-person key sharing with \emph{introductions}, where two users $A$ and $B$ who have not met in person but share a common (trusted) contact $C$ can be introduced to each other by having $C$ send each user a message containing the other's public key. Introductions bootstrap new contacts by re-using the mesh messaging network and existing contacts. This can be thought of as a ``web of trust'' (WoT) approach to offline key distribution, similar to the model adopted by the Briar messaging app~\cite{briar}. Once $A$ and $B$ have each other's public keys, they can initiate introductions with their other contacts. Since individuals wishing to speak to each other at a protest are likely to have met before or to have mutual contacts, introductions can facilitate quick exchange of keys among a large group of users. One potential problem is a malicious user making an incorrect introduction by giving the wrong public key to a user. Since there is no PKI to verify that the correct key has been given, we rely on the WoT. Since only friends can make introductions, introductions are only made by trusted users, i.e. users that shared keys from phase 1, users that shared keys in person, and users that are trusted through the WoT.

To prevent malicious users from sending messages to honest users, the recipient's public key must be known to send them a message. This means that a user that does not know another user's public key must guess, which will be successful with negligible probability. Therefore, malicious users that are not friends with the user requiring an introduction cannot make the introduction. If a user's friend is malicious, the user is vulnerable to this type of attack. So, we allow users to only let certain friends make introductions. Users choosing to trust malicious devices is outside the scope of our threat model. In order to prevent people's behavior changing in response to adversarial messages, we also require that users' devices drop messages intended for them from users not in their friends list, similar to Orca \cite{Orca}. This problem can further be reduced by only allowing introductions through a limited number of hops in the WoT.

Users physically meeting in person before protests begin seems like it has many drawbacks in most settings. This does not allow users to add each other from a distance without having a common connection, requires both users to be available at the same time, and does not take advantage of the existing communications infrastructure that people have access to. In our anonymous mesh setting however, access to the existing communications infrastructure has been removed. This type of key exchange has a number of desirable properties and these drawbacks are less relevant.

We assume that users are attending the same protest and can therefore meet in person. This application is also meant to be used by people that are part of a group that plans to go to a protest. Therefore, all people that a user would message are known ahead of time and they are comfortable meeting in person. By meeting in person and not connecting to services that leave evidence of a user's involvement, users have plausible deniability that they were involved in the protest. Finally, this key exchange method works after Internet access has been cut off. A government may decide to preemptively do this to discourage protesting in the first place, but groups could still self organize, exchange keys, and join larger existing protests even after this occurs. It is also reasonable to assume that the number of people in one of these groups is not too large making this feasible to do.

\subsection{Messaging Phase} \label{sec:active_phase}
Here we describe the Messaging phase using details of the smart broadcasting algorithm discussed in Section \ref{broadcasting_algo}. Note that although we make use of a preparation phase to establish relationships between users, the Messaging phase algorithms we introduce are equally applicable to a system that establishes relationships in another way. This phase can be thought of as a control loop that repeats until the phase ends. An abstract version of this phase using any broadcasting algorithm that fulfills certain requirements is shown in Appendix \ref{sec:active_protesting_appendix}. Pseudo code showing this control loop and related functions is presented in Figure \ref{fig:algorithm}. \ms{PK_s} is the device's public key, \ms{SK_s} is the device's secret key, \ms{PK_r} is the intended receiver's public key, and \ms{PK_{dummy}} is a dummy public key. Each device is independently assigned a uniformly random value indicating when the device is allowed to send a message called \ms{turn}. Devices send messages every \ms{send\_rate} ms from \ms{turn} and all messages are padded to the same length.

% \begin{figure*}
	
% 	\begin{minipage}[t][10cm][t]{0.5\textwidth}
% 		\stackunder{
% 			\stackunder{
% 			    \fbox{
% 					\begin{minipage}[t]{0.9\textwidth}
% 						\input{code/join}
% 					\end{minipage}
% 				}
% 			}{
% 				\fbox{
% 					\begin{minipage}[t]{0.9\textwidth}
% 						\input{code/send}
% 					\end{minipage}
% 				}
% 			}
% 		}{
% 			\fbox{
% 				\begin{minipage}[t]{0.9\textwidth}
% 					\input{code/receive} 
% 				\end{minipage}
% 			}
% 		}{
% 			\fbox{
% 				\begin{minipage}[t]{0.9\textwidth}
% 					\input{code/encrypt} 
% 				\end{minipage}
% 			}
% 		}{
% 			\fbox{
% 				\begin{minipage}[t]{0.9\textwidth}
% 					\input{code/decrypt} 
% 				\end{minipage}
% 			}
% 		}
% 	\end{minipage}
% 	\begin{subfigure}[b]{\textwidth}
% 	\fbox{
% 		\begin{minipage}[t]{0.5\textwidth}
% 			\input{code/main}
% 		\end{minipage}
% 	}
%     \label{fig:main}
% 	\end{subfigure}
% \caption{Main control loop and related functions with details of smart broadcasting.}
% \label{fig:algorithm}
% \end{figure*}

\begin{figure}
	\fbox{
		\begin{minipage}[t]{0.45\textwidth}
			\input{code/join}
		\end{minipage}
	}
	\fbox{
		\begin{minipage}[t]{0.45\textwidth}
			\input{code/send}
		\end{minipage}
	}
	\fbox{
		\begin{minipage}[t]{0.45\textwidth}
			\input{code/receive} 
		\end{minipage}
	}
	\begin{subfigure}[b]{\textwidth}
	\fbox{
		\begin{minipage}[t]{0.45\textwidth}
			\input{code/main}
		\end{minipage}
	}
    \label{fig:main}
	\end{subfigure}
\caption{Main control loop and related functions with details of smart broadcasting.}
\label{fig:algorithm}
\vspace{-2em}
\end{figure}

This control loop consists of checking a series of conditions. Let \ms{DS} be the data structure which stores a record of which messages a user has with corresponding timestamps and having the methods \ms{add} and \ms{remove} to add and remove entries respectively. This is the digest described in \emph{Smart Broadcasting}. \ms{HAS\_SEND}, \ms{OUTSTANDING\_REQUESTS}, \ms{DS\_ARRIVED}, and \ms{MESSAGE\_ARRIVED} are boolean values that are \ms{True} when a device has a message to send, new requests have arrived, a new \ms{DS} has been received, and a new message or request has been received respectively. Otherwise, they are \ms{False}.

First, a device checks if it should send a message by seeing if a multiple of \ms{send\_rate} has passed since \ms{turn}. If it is time to send a message, the device then checks \ms{HAS\_SEND} to see if it should send a real message or a dummy message. If the device sends a message, it adds it to its \ms{DS} and stores the message.

Second, a device checks if the time since \ms{turn} is a multiple of \ms{DS{share}}, the interval at which a device shares its \ms{DS}. If true, the device broadcasts \ms{DS_{compact}}, which is the same as \ms{DS} but contains no timestamps.

Third, if \ms{DS\_ARRIVED} is true, the device saves a list of message ids that the incoming \ms{DS} has that it does not have. This is done by calling \ms{diff}, which is a function that takes in two \ms{DS}s \ms{(DS_1, DS_2)} and returns a list of messages that are in \ms{DS_1} but are not in \ms{DS_2}. The device then broadcasts a request for the messages corresponding to this list of ids.

Fourth, if \ms{OUTSTANDING\_REQUESTS} is true, the device responds to all outstanding requests that it has received from other devices. This is done by taking the union of all lists of requests from other devices, and then broadcasting each corresponding message that the device has in its \ms{DS}. It is assumed that the device has all of the corresponding messages. If the lookup fails due to the message having expired or due to a malicious request, $\varnothing$ is returned from \ms{DS} and will not be broadcast.

Fifth, if \ms{MESSAGE\_ARRIVED} is true, the function \ms{receive} is called. This functions checks if the incoming message is in \ms{DS}. If it is not, it adds it to \ms{DS}, stores it, and tries to decrypt it. If the decryption succeeds, the message and sender are returned.

Finally, the device removes all expired entries from \ms{DS} and their corresponding messages. This is done by checking if the message is older than \ms{time\_to\_keep}, which is a constant value indicating the amount of time that devices store received messages. This is done to prevent the device from having to store irrelevant messages that have already been received by all devices.

This completes the description of the control loop. When a real or dummy message is sent, \ms{send} is called and it encrypts the message, adds the message to its own \ms{DS}, and broadcasts the message. The encryption function is called \ms{E} and the decryption function is called \ms{D}. When a message, \ms{DS}, or request is sent, the function \ms{broadcast} is called and it sends its input message to all of the device's neighbors. The function \ms{join} is called when a device joins the network and is responsible for initializing \ms{turn} and calling the function \ms{init} which creates an empty \ms{DS}. Note that none of these variables need to be synchonized across devices. In fact, having \ms{turn} synchronized across devices would lead to the worst case traffic patterns (see Section \ref{sec:how_much_bandwidth}).

\subsection{Group Messaging}
% Signal group messaging - blog posts and Melissa Chase
% MLS work - Yevgeniy Dodis

Group messaging is an important feature in private messaging and has often been studied in the context where there are centralized servers \cite{alwen_2020, o'leary_2019, chase_2020}. We support group messaging in our protocol by having everyone in a group share a secret key. This key can be shared ahead of time, i.e., everyone in the group decides on a shared key in person before the protest begins, or a group can be built during the protest. In order to build a group after the protest has started, one person can choose a secret key $k$ and send a message to each person in the group containing a special flag to indicate the formation of a group, the secret key for the group, and optionally a list of group members. For a group of size $N$, this would require sending $N$ messages, taking time \ms{send\_rate} * $N$ to establish the group. To speed this up, one can add a list of group members to invite to the group setup message. For example, for Users $A$, $B$, $C$, and $D$, $A$ invites $B$ and also asks $B$ to invite $C$ and $D$. In general, $A$ can invite a small list of users and have each user invite separate lists of other users. In order for users to not be able to impersonate other users in a group, the plaintexts must begin with the sender's signature to prove their identity. Therefore, group messaging can easily be enabled at the cost of setup time and a reduced effective message size.

\subsection{Security Analysis}
We now informally describe why this system meets both our mesh integrity and confidentiality definitions. Appendix \ref{sec:security_games} contains formal proofs of the following theorems.

\begin{theorem}\label{thm:int}
Assuming signcryption scheme $(G,E,D)$ has ciphertext integrity, $\Pi$ has mesh integrity.
\end{theorem}

\begin{theorem}\label{thm:conf}
Assuming signcryption scheme $(G,E,D)$ is a CPA secure cipher and provides public key anonymity, then $\Pi$ has mesh confidentiality.
\end{theorem}

\textbf{Integrity.} If a message is delivered to a device and it decrypts correctly, this means that it was intended for this device and that the message was not tampered with. If it had been tampered with, it would have decrypted to $\bot$ and been discarded due to the guarantees of the signcryption algorithm being used. Thus, our protocol achieves message integrity through the properties of the underlying cipher scheme.

\textbf{Confidentiality.} The actions of a device fall into one of two cases. The first case is the generation of new messages and the consumption of messages from other devices. To an outside observer, the device behaves the same regardless of who it interacts with, what messages it receives, and if it sends a new message. Messages will be generated at a constant rate, where dummy messages are indistinguishable from real messages, and the device will behave the same whether it receives any messages intended for itself or not. In particular, the CPA-security of the signcryption scheme hides message contents, and the anonymity property ensures that each ciphertext reveals nothing about its intended recipient. 

The second case is when a device is exchanging \ms{DS}s with its neighbors and responding to message requests. These actions depend on the behavior of other devices in the protocol. No meaningful metadata can be learned from this process since the encryptions of all messages are considered public. Where a message originated from can be learned by controlling all neighbors of a device and observing a new message from the target device that the devices under your control had not learned about. This is ok since all devices are frequently generating messages whether they are communicating or not, meaning that nothing is learned from this traffic analysis. Which messages a device has learned about can also be observed, but again there is no meaningful information in this as all devices will eventually learn about all messages. When a device receives a particular message tells the attacker nothing about if or who a user is communicating with and what the messages contain.

% \saba{maybe this description can be made shorter/more concise? I added a couple sentences specifically referring to what each security property of the signcryption does for us.}
\section{Implementation}
\label{sec:implementation}

\subsection{Underlying Network}
Our system requires a protocol which allows phones to broadcast messages in a nearby area. Starting in Bluetooth 5.0, a device can broadcast short 255-byte advertisements as part of the pairing process, which are received by nearby devices \cite{collotta_2018}. Our system limits the size of user messages, similar to SMS, to fit within these advertisements. While the theoretical maximum throughput is 2 Mb/s, a reasonable expectation of range is 10m with a throughput of approximately 1.4 Mb/s \cite{collotta_2018}. When obstacles are in the way, this range and throughput may be reduced \cite{yaakop_2017}. In an indoor setting, throughput was reduced to roughly 75\% of its maximum measured throughput (when the devices are right next to each other). This occurred after passing through two walls \cite{collotta_2018} and when using Bluetooth Low Energy \cite{zhang_2020}. Without obstacles, throughput stayed at its maximum past 10m \cite{bocker_2017}. In our analysis we will assume devices have a maximum range of 10 meters (30 ft.), a maximum bandwidth of 1.4 Mb/s, and send messages of length 255 bytes (including headers).

While our analysis will focus on Bluetooth due to its ubiquitous adoption, we note that alternatively using Wi-Fi Direct instead of Bluetooth can dramatically increase our system's capacity and throughput, given that there are measurements showing speeds of several 10s of Mb/s \cite{adam_2020}. Wi-Fi Direct has some similarities to Bluetooth. A device can advertise what services it supports and others can join the network (similar to the advertising and pairing process). Each packet would have around 1500 bytes depending on the MTU, which is significantly larger than the packet constraints of Bluetooth. Using Wi-Fi direct as the underlying network infrastructure instead of Bluetooth would allow the system to scale to much higher levels while also supporting much longer/more frequent messages.

\subsection{Proof of Concept and Model}

To evaluate the performance of our approach on smartphones, we built a proof of concept for Android. We used a 1st generation Google Pixel and a Samsung Galaxy S5 Active. This app encrypts a message and sends it to another phone via a Bluetooth advertisement where it is then decrypted. We also measured the time that decryption operations take on the 1st generation Google Pixel and a 4th generation Google Pixel.

In assessing the scalability of our approach to large protests, we built a model of the devices in the network communicating and running our protocol. This model is a simulation of the network, but we refer to it as a model to avoid confusion with the term simulation also being used in a different context for proofs of security. The traffic patterns that arise as a result of the cover traffic are predictable, as they depend only on how many devices there are, where they are placed relative to each other, and the public parameters. By recording and logging this traffic, average device bandwidth, message delivery rates, how many devices the protocol can support, and other information can be analyzed.
%The purpose of the simulator is not to simulate the internal workings of one device participating in the protocol accurately, i.e. real user messages are not encrypted/decrypted and transferred over implementations of real protocols, but rather the network as a whole with the purpose of studying the previously mentioned list of information.

The model breaks up time into discrete intervals and runs the protocol for each device in each time interval. The interactions between these devices are recorded and the internal state of each device is kept track of. For a more detailed description, see Appendix \ref{sec:sim_details}.

\subsection{Limitations}

We represent a crowd of protesters by placing people in a grid parameterized by the space between neighbors and using this as the underlying graph where an edge exists between any two people that are within the range of Bluetooth. This allows us to investigate the effect of varying how densely packed a crowd is on performance. We do not have more accurate data on how people move during protests, but it is reasonable to assume that people have roughly the same amount of neighbors and that the graph is well connected. If there are sections of the graph that are disconnected or only have one edge connecting them and fail due to too much traffic, this would be equivalent to multiple smaller instances of our protocol running. For very large protests, it is unlikely that the entire graph is connected and we would treat this as multiple smaller protests. Overall, we are capturing the key features of the graph that we are concerned about.
%This is not a perfect model of reality as the topology may have very odd shapes.

We have not evaluated Bluetooth 5 with many phones running in close proximity, and this may cause interference that is not captured by our model. In reality, this may lower the effective bandwidth of each device \cite{zhen_2002}. With that said, it takes a relatively large number of piconets (group of devices sharing a channel) to have the throughput be severely impacted \cite{bocker_2017}. Therefore, this is not a concern in our model excluding potentially the most densely-packed networks.

Phones are not only running our protocol; there are many other apps using system resources, including the possibility of using Bluetooth antennas. If CPU cycles are spent on other tasks, this decreases how many decryptions per second the phone can do and if Bluetooth is being used to transfer data for another application, this lowers how much bandwidth is left for our protocol. Currently, as will be seen in Section \ref{sec:eval}, the CPU is not the bottleneck and is therefore not a major concern. We assume that users of this app do not use other apps that rely on heavy use of the Bluetooth antennas while using this protocol. We also assume that every member of the protest is using this app, which is a very pessimistic assumption because it would lower the spacing between users. Spacing has a direct impact on performance since it affects how many neighbors a device has and therefore how much duplicate traffic it receives. This depends on the choice of the broadcasting algorithm (see Section \ref{broadcasting_algo}). Phones receive all wireless transmissions within range, making sampling of neighbors not an option to reduce this effect. In reality, a fraction of the crowd would be using this protocol, increasing the spacing between users and decreasing the number of total users. Therefore, our model makes reasonable expectations of what conditions can be expected in real settings.

\section{Evaluation}
\label{sec:eval}

Now that we have a model that provides all of the required security and anonymity properties, we can examine how practical it is. How many users does it need to scale to in order to be applied to typical protests in the real world? The median size of a protest in the US was 201 people and the mean was 2824 people. In the UK, the median was 300 people and the mean was 6693 people \cite{Biggs15sizematters}. Small protests are also far more likely to be under reported, making the true mean and median values lower. In Washington D.C., only around 50 protests a year have more than 1000 people \cite{mcphail_2004}. 

While smart broadcasting allows for more users than simple broadcasting by reducing the amount of duplicate data, it does not entirely eliminate it. Two of a device's neighbors may request the same message or two neighbors might both respond to the same missing message request. In both of these cases, only one needed to do so for the original device to receive the information. This creates an interesting effect where \emph{spacing}, how far apart devices are in the protest, impacts a device's bandwidth consumption, i.e. devices having more neighbors increases average bandwidth usage and lowers the capacity of the graph.

Our system can support 200 users when the crowd is tightly packed together (people two to three ft apart) and messages are exchanged frequently (every 30 seconds). When people stand further apart in the crowd (people 15 ft apart) and the message rate is lowered (every minute), our system supports several thousand users. This can be improved further by decreasing the message sending rate and decreasing the strength of Bluetooth antennas in devices based on how many neighbors a user has (see Section \ref{bt_range_limit}). Almost all messages are delivered within a few seconds in smaller, crowded topologies (a few hundred users standing a few ft apart) and within two minutes in much larger topologies (several thousand users standing 15 ft apart). Our system performs well enough to accommodate protests of moderate size, providing reliable message delivery in a reasonable amount of time. For very large protests (several tens of thousands of people or more), it is likely that these crowds would not form a fully connected graph and can be thought of as many instances of these moderately sized protests, making our system a potential option for these environments as well.

The model is written in Python 3 and was run on a Macbook Pro (2019) and machines on Google Cloud Compute Engine running Ubuntu 20.04. The specs of the host machines do not affect the results of the model. The logs that are produced can be quite large (10s to 100s of GB) and are also processed in Python. These logs are only used for analyzing our model and would not be produced in a deployed system.

Now, we examine if the following performance requirements are met: 
\vspace{-1ex}
\begin{myitemize}
    \item Device computations, which primarily are comprised of cryptographic operations, are not the bottleneck
    \item Capacity is large enough for real protests
    \item Messages are delivered in a timely matter
\end{myitemize}
\vspace{-1ex}
Specifically, we answer the following questions: How much bandwidth does a device use? How many devices can this design support? How quickly are messages delivered? How do factors such as how many devices are participating, how close people are together, how often a \ms{DS} is shared, and how does the message sending rate affect these questions? Finally, we demonstrate the performance benefits that smart broadcasting has over simple broadcasting.

\subsection{Is Cryptography the Bottleneck?}

\begin{table}
\centering
\begin{tabular}{|l|l|l|}
\hline
\textbf{Phone} & \textbf{Dec. Time} & \textbf{Throughput} \\ \hline
Pixel (1st gen) & 0.57ms & 3.60 Mbps \\
Pixel (4th gen) & 1.35ms & 1.52 Mbps \\ \hline
\end{tabular}
\caption{Average decryption time on Android devices for Bluetooth messages and the corresponding maximum throughput that could be achieved without the decryptions being the bottleneck.}
\label{table:decryption-results}
\vspace{-3em}
\end{table}

Every phone needs to check whether it is a recipient for every message sent in the system. As described in Section~\ref{sec:system-description} this requires doing a decryption using ElGamal and AES. If this were too slow, it would limit the system's capacity.

We ran microbenchmarks to evaluate whether decryptions are a bottleneck and found that they are not. We used implementations of ElGamal over the elliptic curve p256v1 \cite{lubux2020} and AES written in C, and ran them on a Google Pixel (1st Generation) and a Google Pixel (4th Generation). We generated one thousand random messages of the same length used in our protocol. We encrypted and decrypted each string, and timed the results. Using the average decryption time, we calculated how much data each phone would be able to decrypt per second and defined this as \emph{Throughput}. Given that a phone can only receive up to 1.4 Mbps of data over Bluetooth, any Throughput higher than 1.4 Mbps implies that all messages can be decrypted at least as quickly as they are received. Table~\ref{table:decryption-results} shows the results for both phones. Since cryptographic operations are not the bottleneck of the system, we do not incorporate them in our model.

\subsection{How many devices can this design support?}
To figure out the maximum capacity of the network, we found the maximum number of devices that could be placed into the network while keeping the 95th percentile of bandwidth under 1.4 Mb/s for various spacings, send rates, and broadcasting types (Figure \ref{fig:capacities}). Realistically, the topologies will have high degrees and many paths between two randomly chosen devices. This means that messages still will typically have multiple delivery paths if a higher number of devices ``fail'' by having temporary bursts above 1.4 Mb/s. The 90th and 85th percentiles are graphed as a result and scale to several thousand devices at higher spacings. Lowering message frequency increases the capacities even further at all spacings.

\begin{figure*}
     \centering
     \begin{subfigure}[b]{0.3\textwidth}
         \centering
    	 \includegraphics[width=\linewidth]{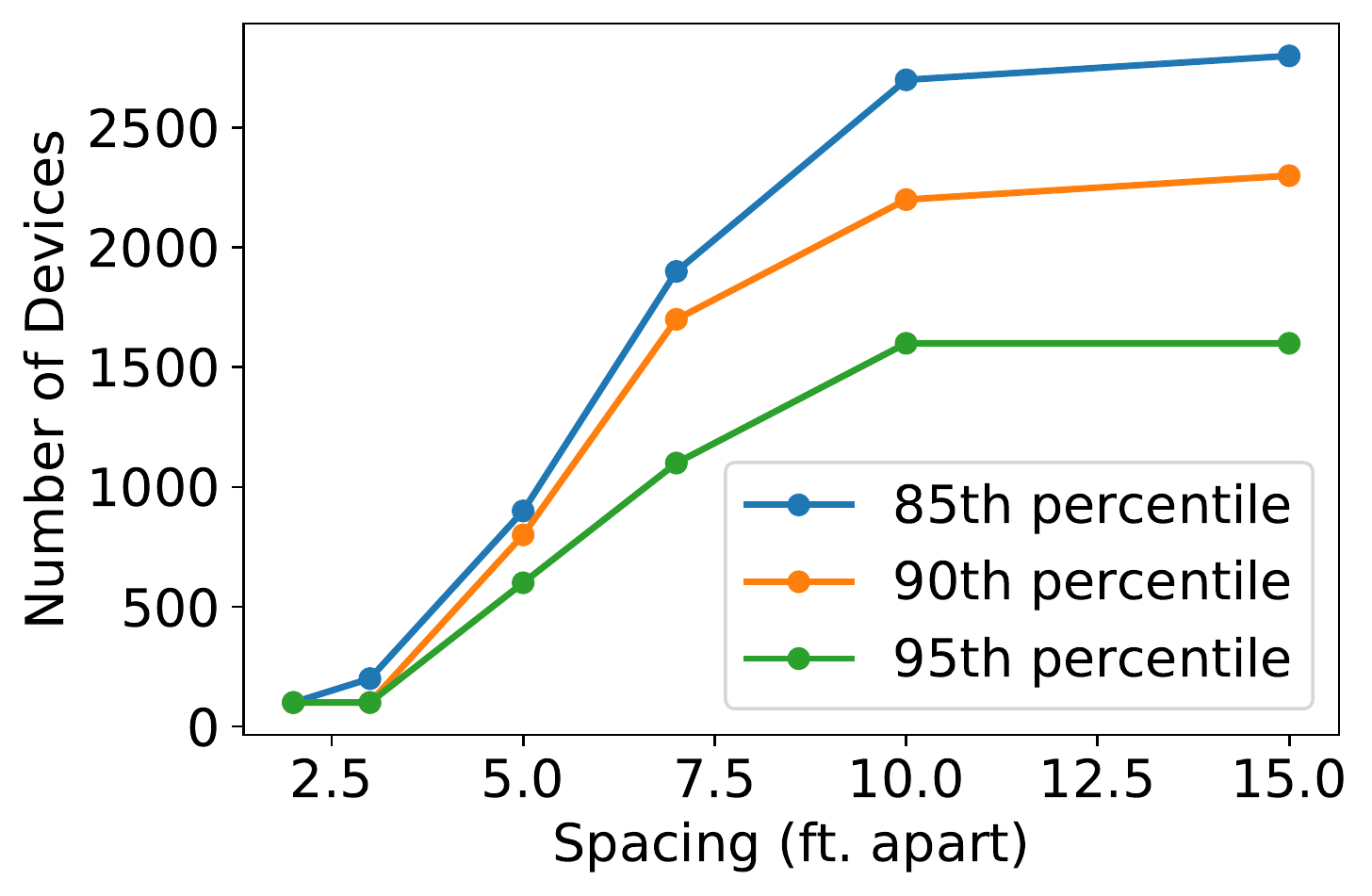}
         \caption{Smart broadcasting - 30 second message interval}
         \label{fig:capacities_smart_30}
     \end{subfigure}
%     \hfill
     \begin{subfigure}[b]{0.3\textwidth}
         \centering
    	 \includegraphics[width=\linewidth]{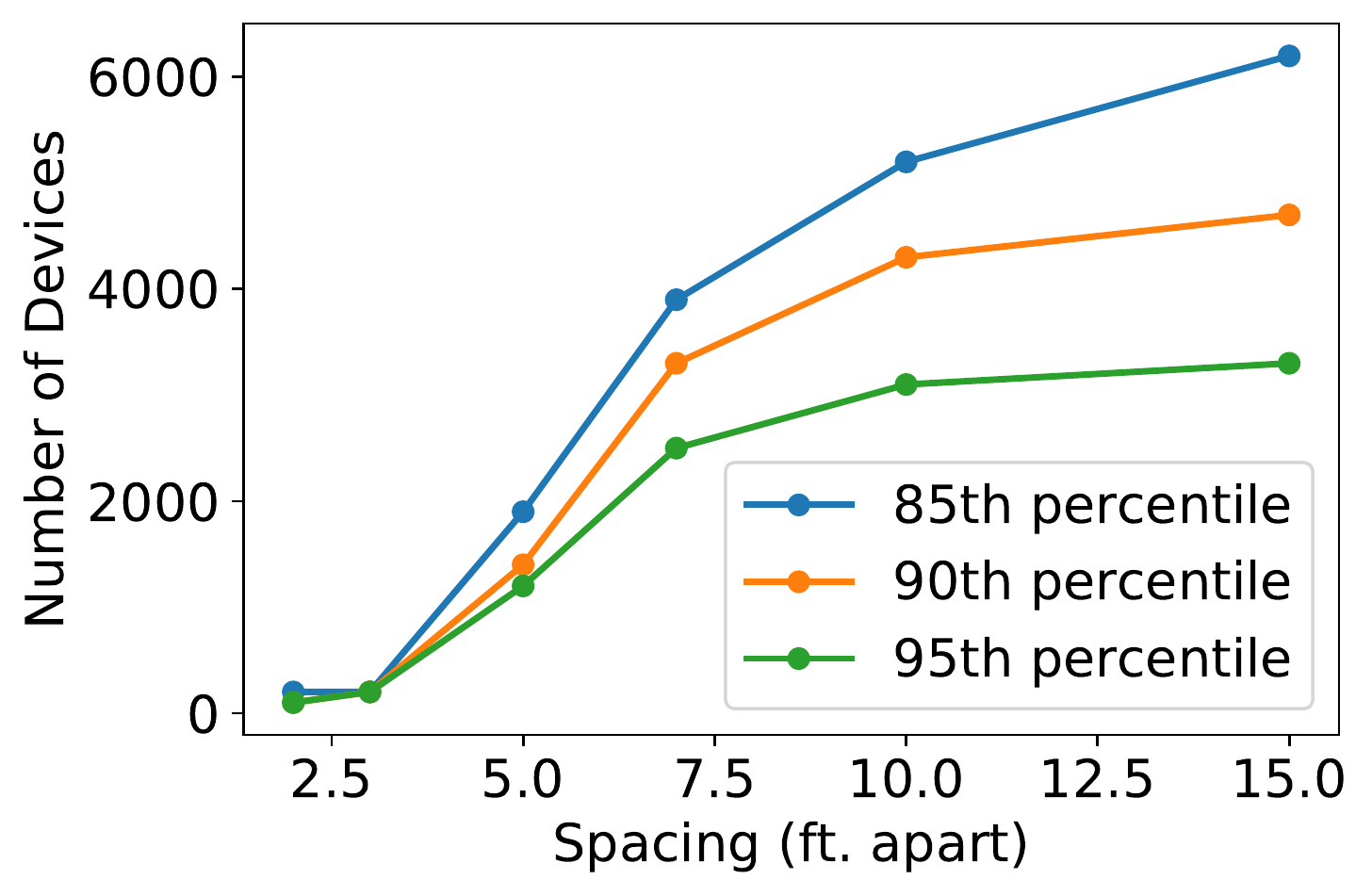}
         \caption{Smart broadcasting - 60 second message interval}
         \label{fig:capacities_smart_60}
     \end{subfigure}
%     \hfill
     \begin{subfigure}[b]{0.3\textwidth}
         \centering
    	 \includegraphics[width=\linewidth]{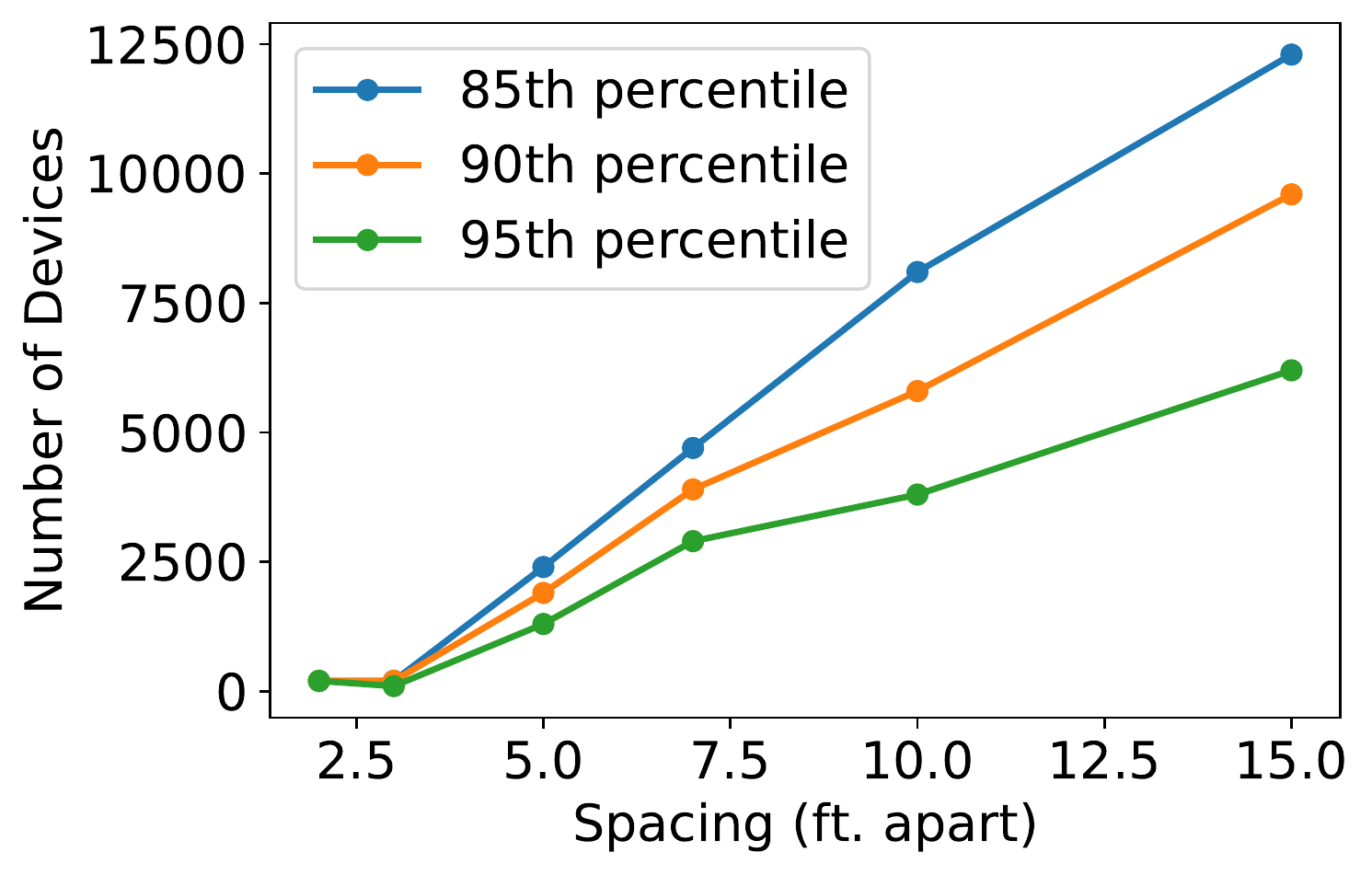}
         \caption{Smart broadcasting - 120 second message interval}
         \label{fig:capacities_smart_120}
     \end{subfigure}
%     \hfill
    %  \begin{subfigure}[b]{0.4\textwidth}
    %      \centering
    % 	 \includegraphics[width=\linewidth]{capacity_minutes_5_bitListInterval_10000_move_False_uniform_True_broadcastType_simple_sendRate_30000.pdf}
    %      \caption{Simple broadcasting - 30 second interval}
    %      \label{fig:capacities_simple_30}
    %  \end{subfigure}
     \caption{Capacities for smart broadcasting with various message sending intervals. All models ran for 5 minutes. The number of devices is the capacity of the system rounded down to the nearest 100. Depending on the spacing of the crowd and message sending interval, our system supports 200 to over 5000 devices.}
     \label{fig:capacities}
     \vspace{-1em}
\end{figure*}

\subsection{How Much Bandwidth Does a Device Use?} \label{sec:how_much_bandwidth}
\begin{figure*}
     \centering
    %  \begin{subfigure}[b]{0.3\textwidth}
    %      \centering
    % 	 \includegraphics[width=\linewidth]{bandwidth_minutes_5_bitListInterval_10000_move_False_uniform_True_broadcastType_smart_uniformSpacing_2_sendRate_30000}
    %      \caption{2 ft spacing delivery}
    %      \label{fig:2_bandwidth}
    %  \end{subfigure}
    %  \hfill
     \begin{subfigure}[b]{0.27\textwidth}
         \centering
    	 \includegraphics[width=\linewidth]{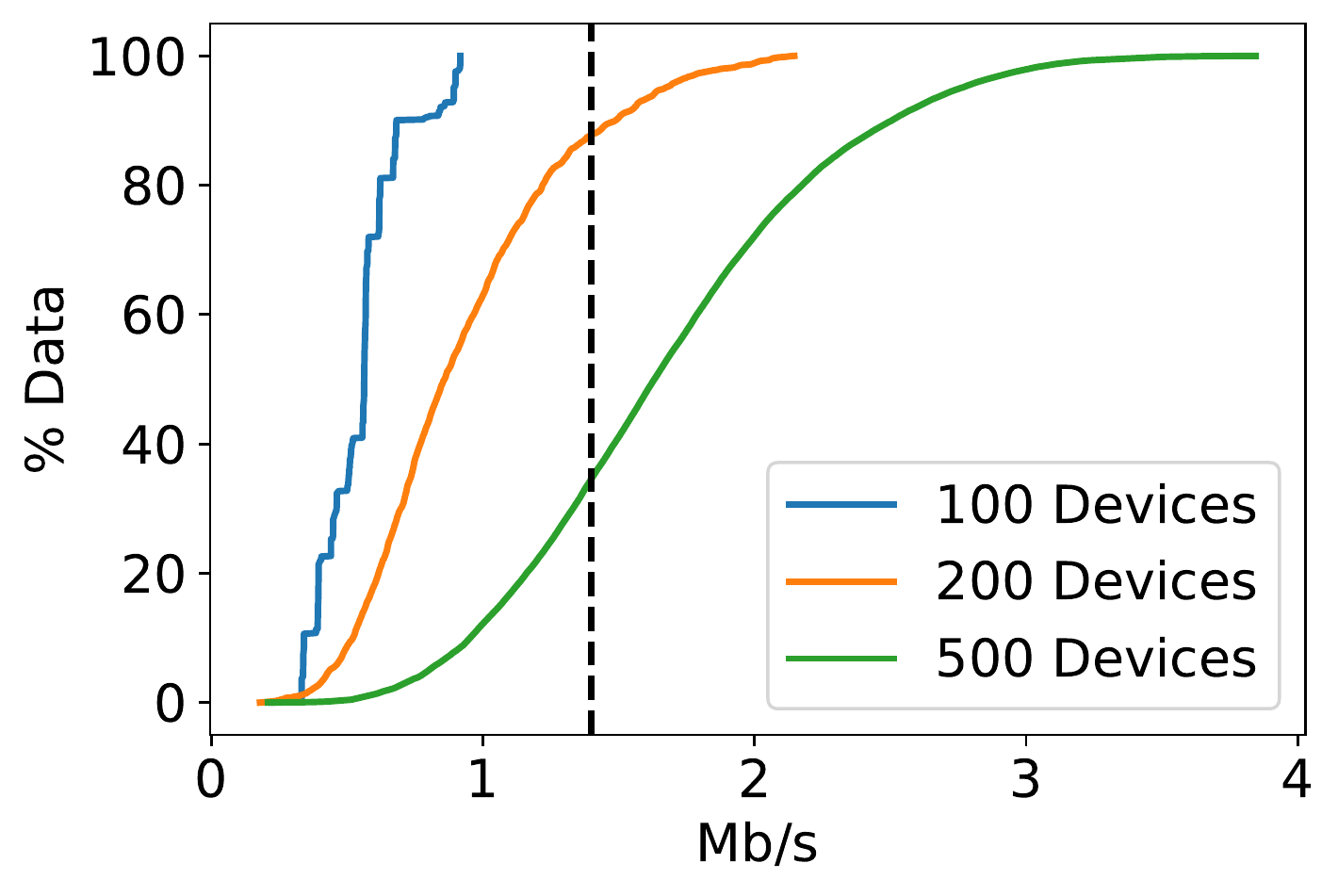}
         \caption{3 ft spacing delivery}
         \label{fig:3_bandwidth}
     \end{subfigure}
     \hfill
    %  \begin{subfigure}[b]{0.3\textwidth}
    %      \centering
    % 	 \includegraphics[width=\linewidth]{bandwidth_minutes_5_bitListInterval_10000_move_False_uniform_True_broadcastType_smart_uniformSpacing_5_sendRate_30000}
    %      \caption{5 ft spacing delivery}
    %      \label{fig:5_bandwidth}
    %  \end{subfigure}
    %  \hfill
     \begin{subfigure}[b]{0.27\textwidth}
         \centering
    	 \includegraphics[width=\linewidth]{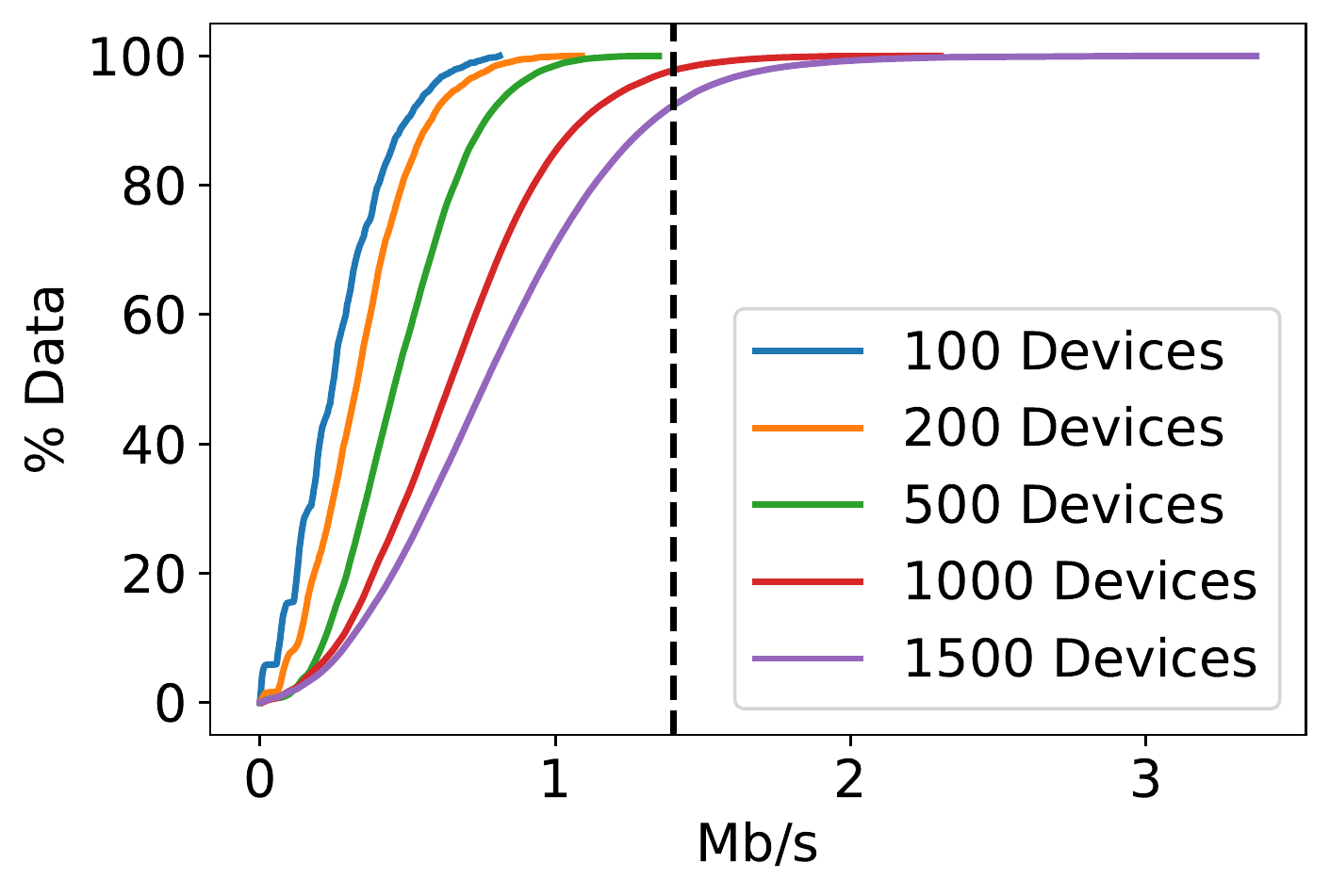}
         \caption{7 ft spacing delivery}
         \label{fig:7_bandwidth}
     \end{subfigure}
     \hfill
    %  \begin{subfigure}[b]{0.3\textwidth}
    %      \centering
    % 	 \includegraphics[width=\linewidth]{bandwidth_minutes_5_bitListInterval_10000_move_False_uniform_True_broadcastType_smart_uniformSpacing_10_sendRate_30000}
    %      \caption{10 ft spacing delivery}
    %      \label{fig:10_bandwidth}
    %  \end{subfigure}
    %  \hfill
     \begin{subfigure}[b]{0.27\textwidth}
         \centering
    	 \includegraphics[width=\linewidth]{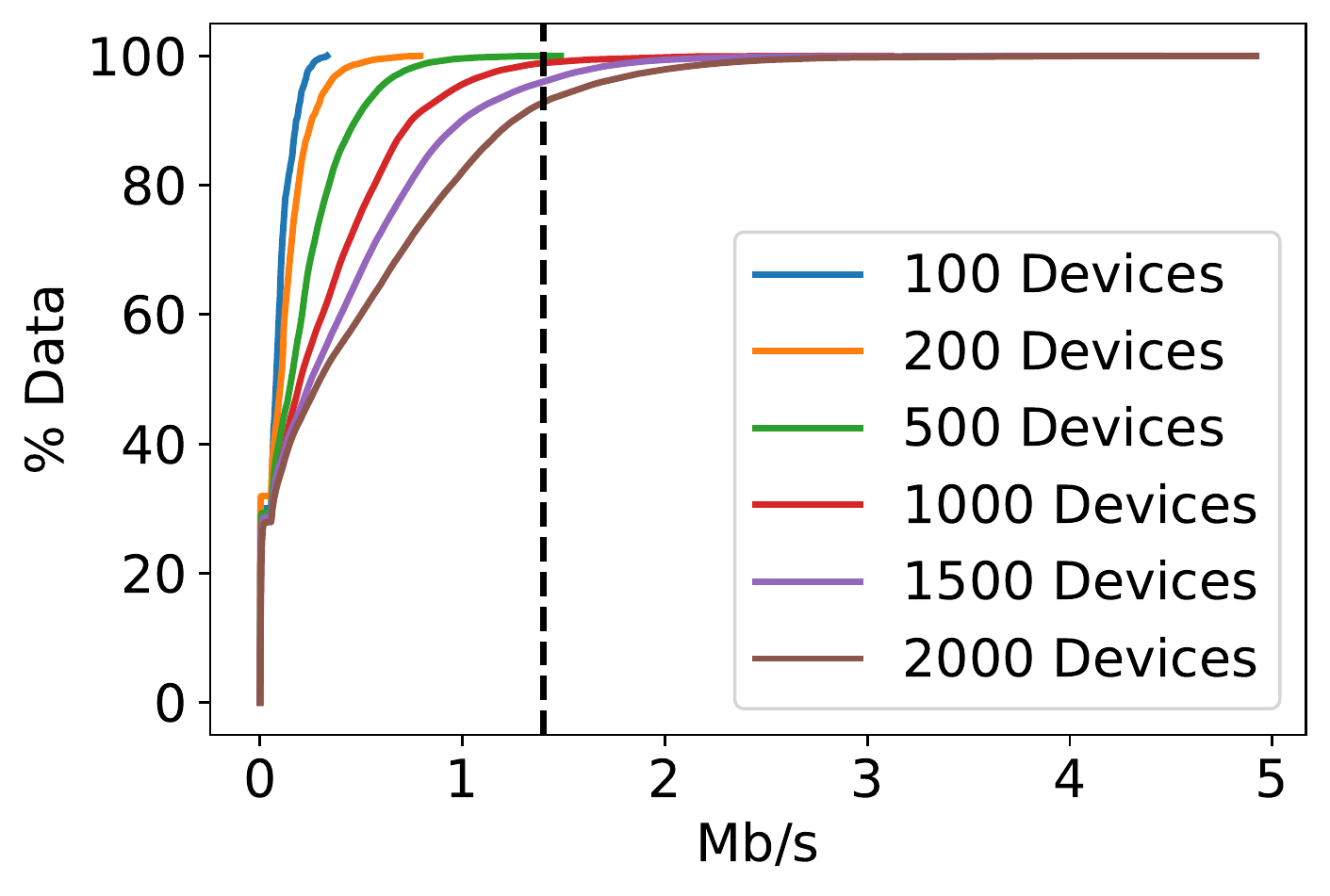}
         \caption{15 ft spacing delivery}
         \label{fig:15_bandwidth}
     \end{subfigure}
	 \caption{Bandwidth CDFs for various spacings. All models ran for 5 minutes, had a \ms{DS} exchange interval of 10 seconds, and sent messages every 30 seconds. While at lower spacings a large portion of devices exceed their bandwidth limits after 200 devices participate, at higher spacings most devices stay below their bandwidth limits even after the number of devices in the network increases to over 1000.}
     \label{fig:bandwidths}
     \vspace{-1em}
\end{figure*}

This depends on how many devices are in the network, how many neighbors a device has, how frequently \ms{DS}s are exchanged, and how frequently messages are sent. Since all devices receive all messages from all other devices, bandwidth per device scales at least linearly with the number of participants in the network. How bursty the traffic is and the topology affects bandwidth as well. If all devices sent their messages at the same time, this would be the worst case scenario. We send messages at intervals relative to the uniform random value \ms{turn} to try and avoid this. Similarly, if the graph had the topology of a barbell (where half of the graph needs a specific edge in all of its paths to the other half of the graph), the devices that form the connecting edge would receive a lot of traffic and potentially fail. 

The number of neighbors a device has affects a device's bandwidth due to the broadcast communication channel. Since devices have a fixed range that they can broadcast/receive messages in, spacing has a large impact on performance. The closer people are, the more neighbors they have, which increases the degree of the graph. While there are a fixed number of messages that are possible to hear, having more neighbors means that a device will hear more people's requests for missing messages, others' responses, and more exchanges of \ms{DS}s. If simple broadcasting was used, more duplicate messages would be heard as well. An example is when devices to the left and right of a user are responding to requests for the same message and the user hears both of these responses due to the nature of broadcast communication. Figure \ref{fig:bandwidths} shows the effects of spacing when all other factors are held constant. When people are tightly packed together, after a few hundred people join the network, large portions of the CDF fall above 1.4 Mb/s while it takes thousands to see similar effects when people are 15 ft. apart.

\begin{figure}
     \centering
     \begin{subfigure}[b]{0.35\textwidth}
         \centering
    	 \includegraphics[width=\linewidth]{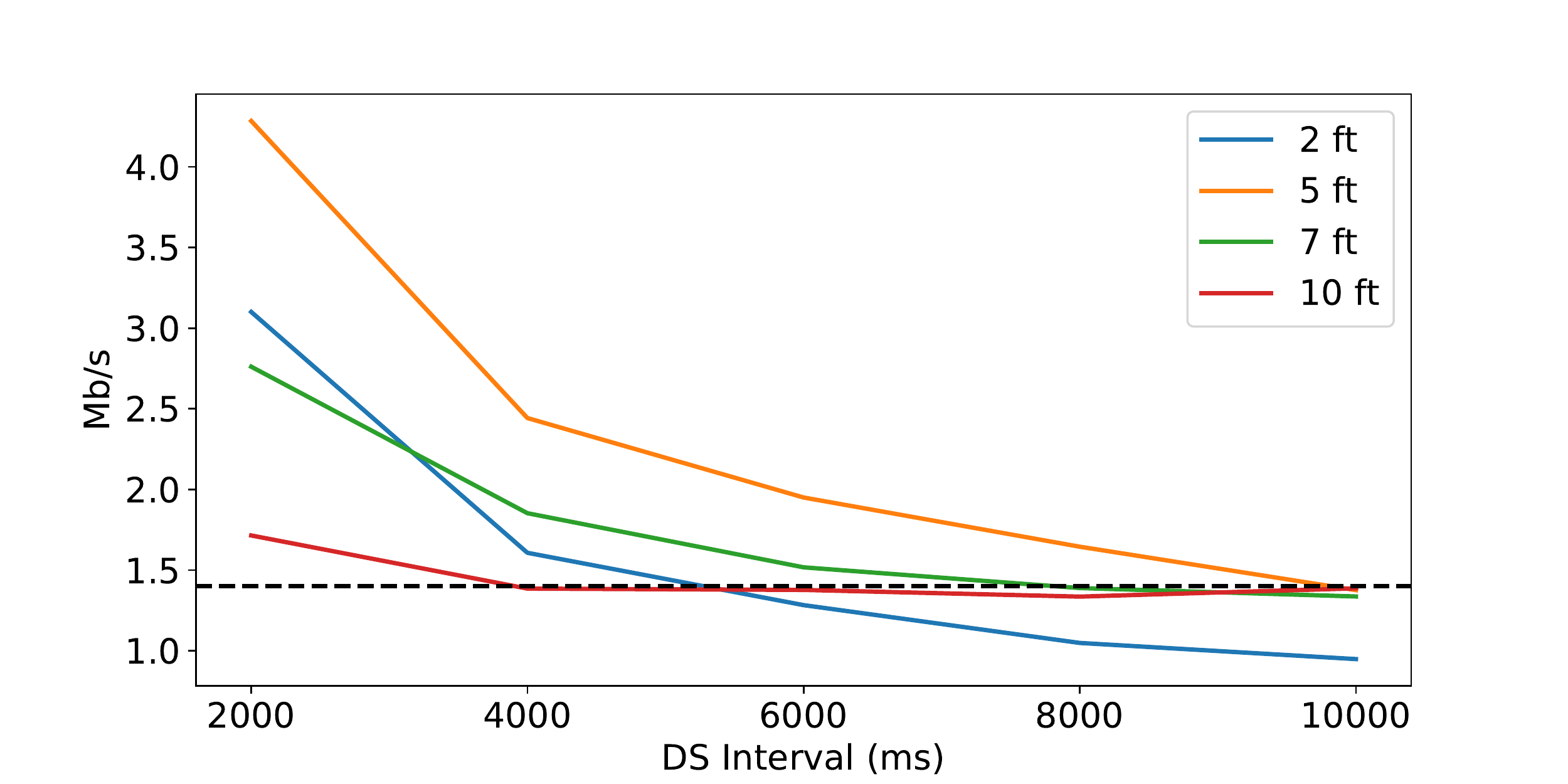}
         \caption{95th Percentile Bandwidth for various \ms{DS} exchange intervals at lower spacings. All slopes are negative, but the slope starts to flatten out at a spacing of 10 ft.}
         \label{fig:ds_low}
     \end{subfigure}
     \hfill
    \begin{subfigure}[b]{0.35\textwidth}
         \centering
    	 \includegraphics[width=\linewidth]{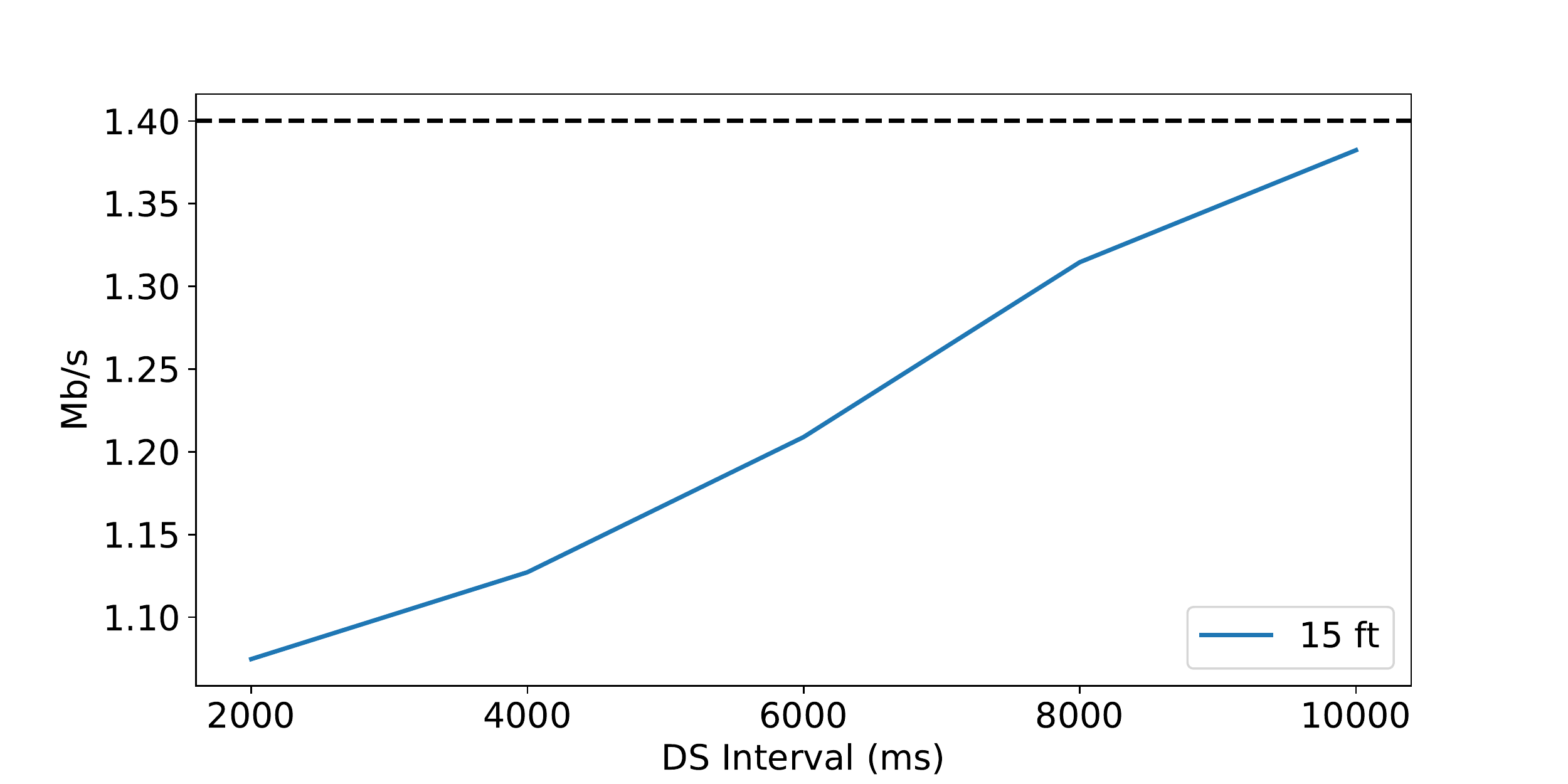}
         \caption{95th Percentile Bandwidth for various \ms{DS} exchange intervals at a higher spacing of 15 ft. The slope is now positive.}
         \label{fig:ds_high}
     \end{subfigure}
     \hfill
        \caption{All models ran for 5 minutes and sent messages every 30 seconds. At lower spacings, increasing the \ms{DS} exchange interval lowers the bandwidth used by devices, increasing the capacity of the system. At higher spacings, increasing the \ms{DS} exchange interval increases the bandwidth used by devices, decreasing the capacity of the system. Therefore, the optimal \ms{DS} exchange interval depends on the spacing in the protest.}
        \label{fig:bitlist_effects}
        \vspace{-1em}
\end{figure}

The rate at which \ms{DS}s are exchanged has an effect as well. If \ms{DS}s are exchanged too frequently, then more bandwidth will be used exchanging \ms{DS}s that will likely contain little to no new information. If \ms{DS}s are exchanged too infrequently, then devices learn about many missing messages all at once. The responses from the resulting requests will be larger as more messages need to be shared and therefore traffic will be more bursty. This is further complicated by the fact that the optimal interval also depends on the spacing of devices. Topologies that have devices close together benefit from exchanging their \ms{DS}s less frequently since they have more neighbors to hear from and the state between neighbors is similar. Topologies that have devices more spread out benefit from doing so more frequently. This is because they have fewer neighbors and when they hear an update, it is more likely to contain a larger amount of new information. Also, since there are less neighbors, this decreases bandwidth from updates, holding all else equal (see Figures \ref{fig:ds_low} \& \ref{fig:ds_high}).

Finally, message frequency has an intuitive effect on bandwidth. Increasing the frequency that messages are sent at results in each device sending more messages. As a consequence, all devices have to process and forward more messages as well, increasing all of their bandwidth usage. This increase is directly proportional to the previous bandwidth usage before the increase, excluding non-message traffic such as exchanging \ms{DS}s.

\subsection{How quickly are messages delivered?}
\begin{figure}
    \centering
    \includegraphics[width=0.6\linewidth]{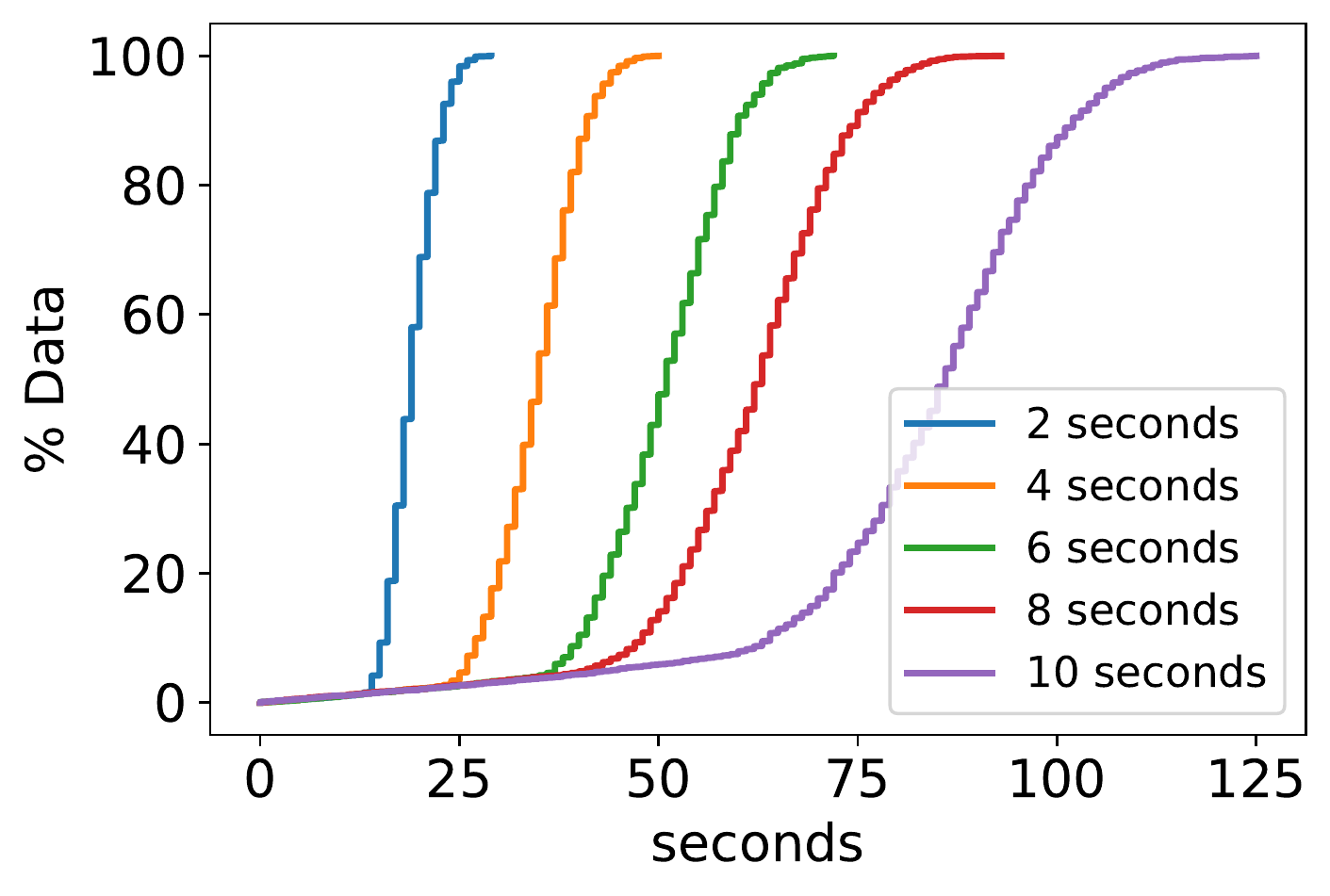}
    \caption{Delivery CDF for various \ms{DS} exchange intervals. All models ran for 5 minutes, had a spacing of 15 ft., and sent messages every 30 seconds. Increasing the \ms{DS} exchange interval causes longer delivery time for messages because each message must ``wait'' longer before traversing each hop in the network. This can cause a trade-off between capacity and message delivery time for certain network configurations.}
    \label{fig:bitlist_delivery}
    % \vspace{-4em}
\end{figure}
\emph{Message delivery time} is defined as the time that it takes for a message to reach all other devices in the graph from the point in time that it was created. This is a very conservative definition as these delivery times have somewhat long tailed distributions (Figure \ref{fig:deliveries}) and many phones will receive the message before the message is considered ``delivered.'' Adding devices to the graph may increase its diameter which corresponds to longer delivery times as messages need to traverse more hops to be delivered. Spacing impacts this as well since the diameter of the graph may shrink as people move closer together. Increasing the \ms{DS} Exchange Interval size increases the message delivery time (Figure \ref{fig:bitlist_delivery}) since each message effectively has a longer waiting period before traversing the next hop. In smaller topologies, messages are delivered within a few seconds while in larger topologies with devices spaced out 15 ft apart, messages are delivered within 2 minutes in the model. It is important that message delivery not take longer than the \ms{time\_to\_keep} in the \ms{DS} as this could result in infinite loops. This means that messages are delivered quickly enough to support a wide range of communication patterns including notifying friends of a change in location, finding friends after becoming separated, informing friends of an important event at a distant part of the protest, etc.

\begin{figure*}
     \centering
     \begin{subfigure}[b]{0.25\textwidth}
        %  \centering
    	 \includegraphics[width=\linewidth]{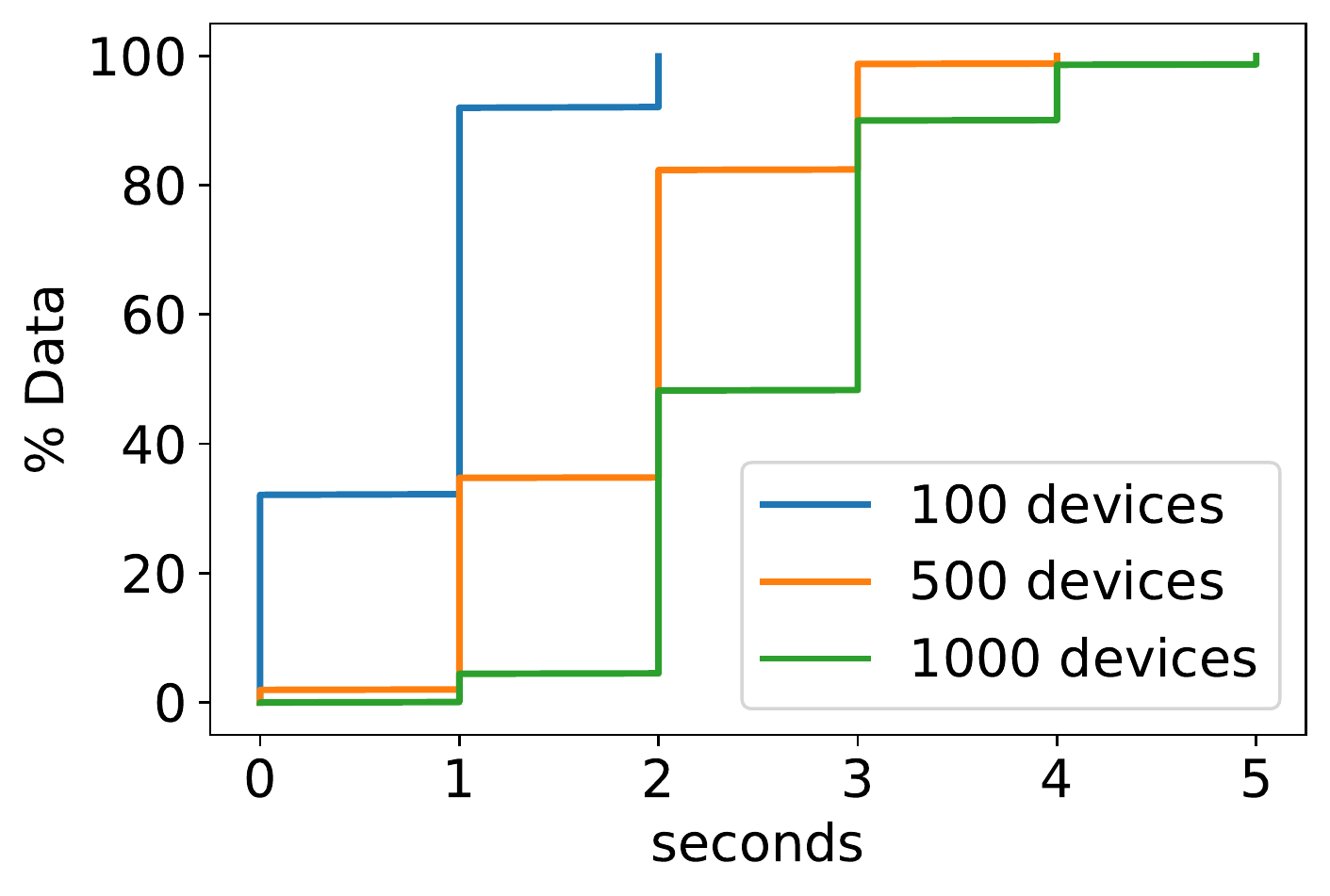}
         \caption{5 ft spacing delivery}
         \label{fig:5_delivery}
     \end{subfigure}
%     \hfill
    %  \begin{subfigure}[b]{0.38\textwidth}
    %     %  \centering
    %      \includegraphics[width=\linewidth]{delivery_minutes_5_bitListInterval_10000_move_False_uniform_True_uniformSpacing_7_sendRate_30000.pdf}
    %      \caption{7 ft spacing delivery}
    %      \label{fig:7_delivery}
    %  \end{subfigure}
%     \hfill
     \begin{subfigure}[b]{0.25\textwidth}
        %  \centering
         \includegraphics[width=\linewidth]{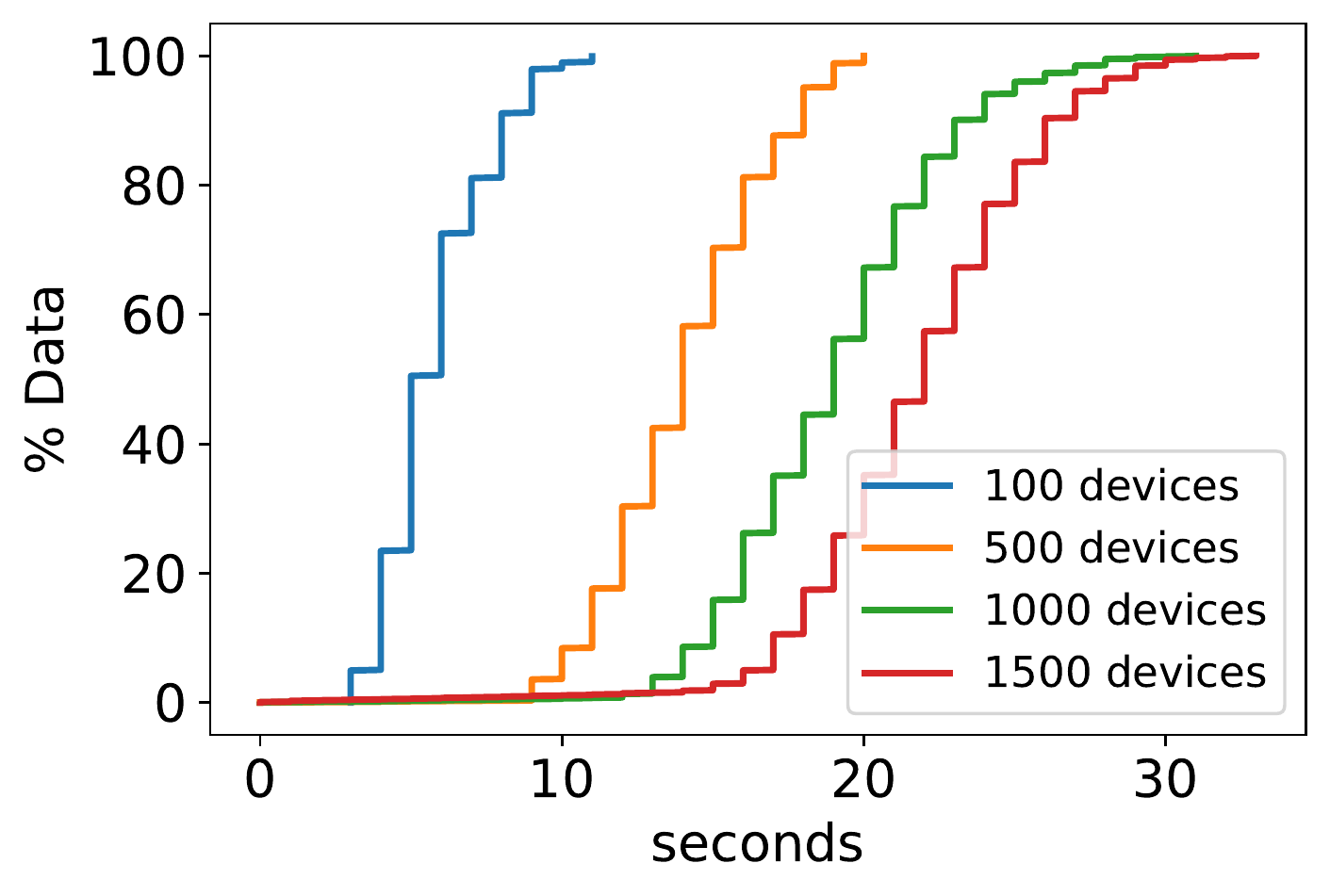}
         \caption{10 ft spacing delivery}
         \label{fig:10_delivery}
     \end{subfigure}
%     \hfill
     \begin{subfigure}[b]{0.25\textwidth}
        %  \centering
         \includegraphics[width=\linewidth]{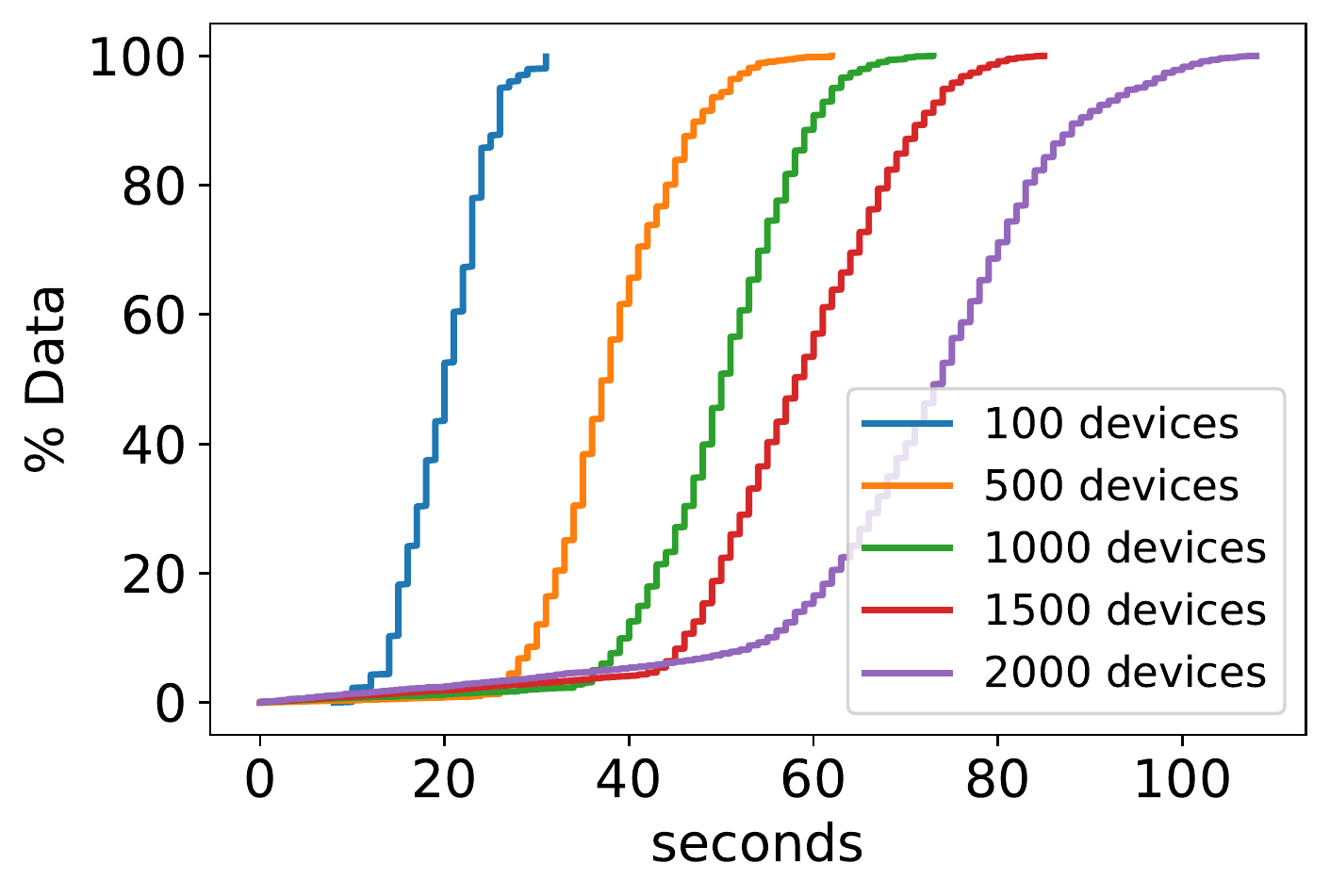}
         \caption{15 ft spacing delivery}
         \label{fig:15_delivery}
     \end{subfigure}
    \caption{Delivery Time CDFs for various spacings. All models ran for 5 minutes, had a \ms{DS} exchange interval of 10 seconds, and sent messages every 30 seconds. Increasing the spacing causes longer delivery times for messages because the graph is ``less connected,'' requiring messages to go through more hops. This causes a trade-off between capacity and message delivery time.}
    \label{fig:deliveries}
    % \vspace{-2em}
\end{figure*}

\subsection{Benefits of Smart Broadcasting}
\begin{figure}
\centering
\includegraphics[width=0.6\linewidth]{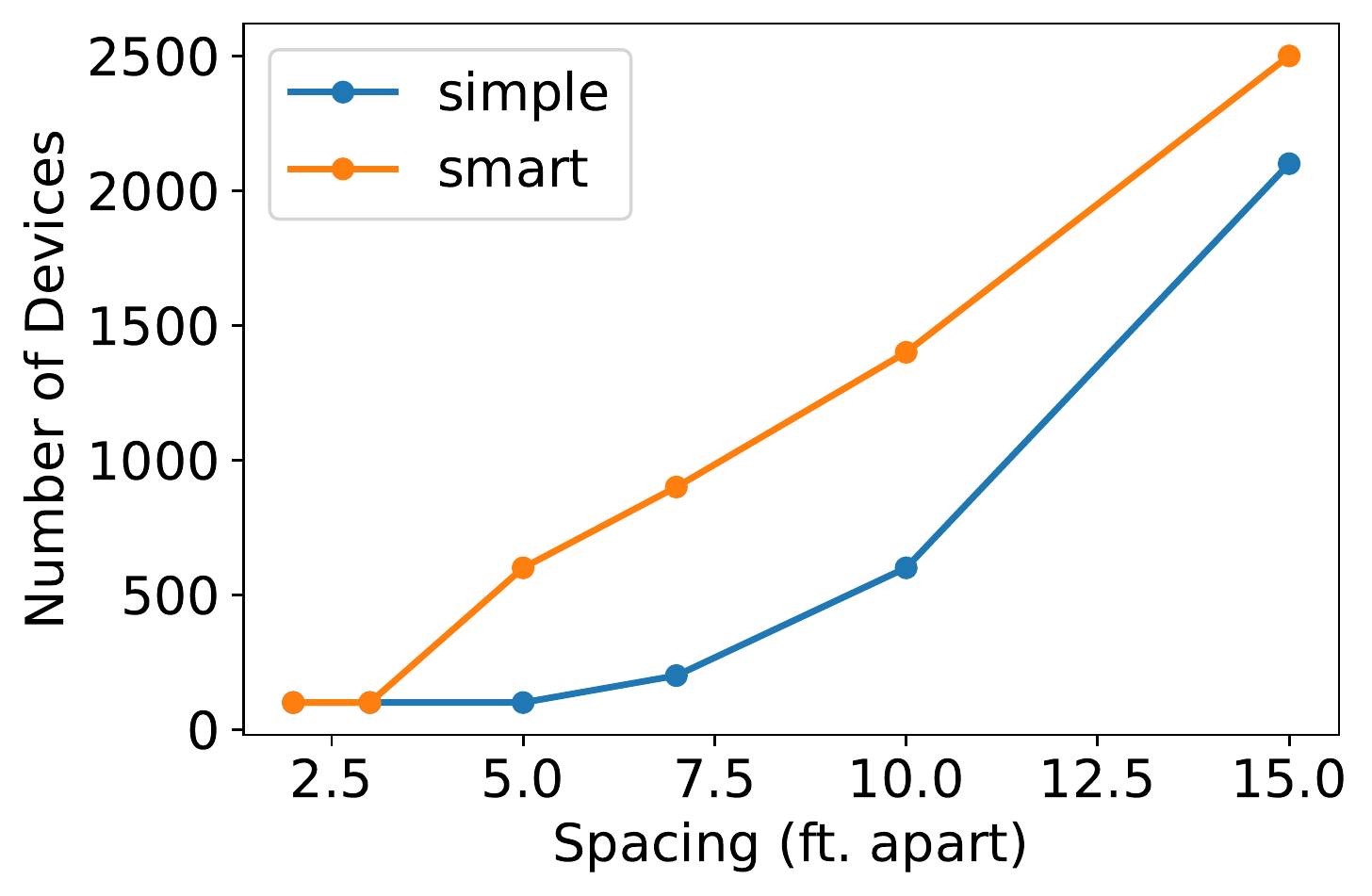}
\caption{Comparison of simple and smart broadcasting with a varying \ms{DS} exchange rate of 10s (for spacings less than 15 ft) and 2s (for spacing of 15 ft). These values were chosen based on the results of Figure \ref{fig:bitlist_effects}. Spacing can be approximated by a device by using the number of neighbors it has. The 95th percentile of bandwidth is used to find the capacity, all models ran for 5 minutes, and messages are sent every 30 seconds. Smart broadcasting performs at least as well (often substantially better) than simple broadcasting at all spacings.}
\label{fig:simple_vs_smart}
\vspace{-1em}
\end{figure}

In addition to the reliable delivery properties that smart broadcasting provides, it also allows for higher capacities at most spacings. Smart broadcasting outperforms simple broadcasting at all tested spacings higher than three feet. The two algorithms also perform equally well at lower spacings. Figure \ref{fig:simple_vs_smart} shows these results and the following parameters were used: messages were sent every 30 seconds, the 95th percentile of bandwidth was used to calculate the capacity, and smart broadcasting had a \ms{DS} exchange rate that varied between two and 10 seconds. The optimal \ms{DS} exchange rate is a function of the spacing of the graph, which can be calculated from the number of neighbors that a device has. As the spacing increases, devices have less neighbors, meaning that they receive less digests and message requests from their neighbors. Asymptotically, this becomes closer to simple broadcasting as spacing increases, and the two algorithms' performance will converge.
\section{Design Extensions}
\label{sec:extensions}

Here we discuss further extensions to our system. We show how to reduce the negative effects of spacing by changing the signal strength of Bluetooth antennas. For networks with higher available bandwidth where the cost of decryptions may become the bottleneck, we describe a solution where the decryption cost scales with the number of friends that a user has. Finally, for very large protests, we outline how one can achieve substantially further scaling at the cost of some anonymity.

\subsection{Limiting Bluetooth Range} \label{bt_range_limit}
Decreasing the strength of the Bluetooth antennas in devices would not only decrease power consumption, but can also increase the capacity of networks with tightly packed crowds. Artificially lowering the range of Bluetooth devices would be equivalent to increasing the spacing of devices, which would increase the capacity of the network. For example, this can make a network with a spacing of 5 ft have a capacity closer to one with a spacing of 15 ft. Automatically adjusting the strength of the antennas based on how many neighbors a device has can alleviate constraints of networks with low spacings and avoid trade-offs between spacing and the \ms{DS} exchange rate. When advertising a message in Android, one of four power levels can be specified. We measured the range of these settings on a first generation Google Pixel and a Samsung Galaxy S5 Active in an outside area with direct lines of site between the phones and no other Bluetooth devices providing interference. At the lowest power setting, we found that the Samsung phone received advertisements from the Pixel at a distance of 40 ft. At the highest power setting, this extended past 130 ft. Given objects blocking direct lines of site, e.g., walls, buildings, other people, and the interference of other Bluetooth devices, these ranges would be greatly reduced in the setting of a real protest and we are more interested in the ratio of these numbers.

\subsection{Symmetric vs Public Key Decryption}
One significant computational cost of this system is each device trying to decrypt each message that it receives. This could potentially lead to denial of service (DoS) attacks where a device has to process many messages quickly and cannot do it in time, having to drop them. One potential solution to this is for device A sending a message to device B to MAC their messages with a key shared between A and B. Then, B can use symmetric cryptography to quickly check if a message is for them from A. Currently, this is not the bottleneck for scaling or the most serious denial of service attack, as the network links would fail before devices fail due to computation limits. However for completeness, we discuss the trade-offs of this approach in this section.

MACing messages with pre-shared secret keys has two challenges. The first is sharing the key. Our protocol allows users to share public keys in person during the \emph{Key Distribution Phase} (see Section \ref{sec:key_distribution}). These shared MAC keys can also be exchanged at this point, solving the first problem. The second problem is that a device must check the MAC with each shared MAC key that it possesses to determine if the message is for itself. Let \ms{Friends} be the list of friends that a device has and \ms{MAC_{time}} be the time it takes to verify a MAC. Since the device has one shared key for each friend, this requires \ms{\lvert  Friends \rvert} MAC checks. Therefore, which is faster depends on the time to compute \ms{\lvert Friends \rvert * MAC_{time}} and  the time it takes to do a public key decryption.

\subsection{Clique Broadcasting}
\begin{figure}
    \centering
    \includegraphics[width=.6\linewidth]{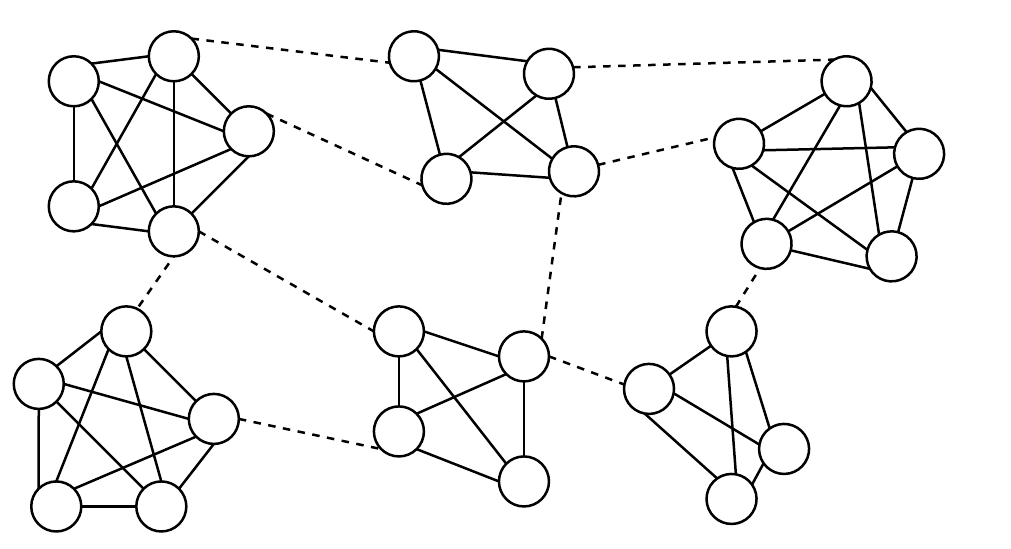}
    \caption{Graph over cliques. Connections within a clique are shown with a solid line. The logical connections that form the graph over the cliques are shown with a dashed line.}
    \label{fig:cliques}
    \vspace{-1em}
\end{figure}

Further capacity scaling can be gained with a more sophisticated broadcasting algorithm. The idea is to make a graph on top of the graph called the ``clique graph,'' where each node in the clique graph is a group of nodes in the original network that have formed a clique. All nodes in a group must be in a clique so that when messages are broadcast within a clique, no information can be learned from the transmission time or the number of hops taken. Each group has a leader who is in charge of choosing which edges to use to send a message from their clique to a neighboring clique (in the original system, there might be multiple edges going between these groups and a message may be sent over more than one of these edges). This decreases duplicate traffic. When a device receives a message from another clique, they broadcast it within their clique so that all members have the message. The leader then chooses how to rebroadcast it to neighboring cliques, or not to if they already have the message which can be done by exchanging \ms{DS}s (smart broadcasting) or not rebroadcasting messages that have already been seen by the clique (simple broadcasting). This allows the protocol to roughly scale by a factor of the clique size by cutting down on duplicate messages being shared between devices in separate cliques, allowing for much larger protests. Forming these cliques may be difficult, but one potential solution is to greedily form them. If a device does not end up in one, it can be its own clique (size of 1) and move around the physical environment until it can join a larger one.

If every group also chooses a group ID and users notify their friends of their group ID (through the non-clique broadcasting), then messages can be encrypted by the recipient's key and then their group's key. Groups can choose to not broadcast messages within their clique and just route it through to the next clique if it is not for their group. This allows for a further degree of scaling but sacrifices some anonymity.

\section{Related Work}

\subsection{Bridgefy Security}

Albrecht~\etal~\cite{breakingbridgefy} gives a high level overview of the apps that exist in this space and reverse engineers Bridgefy, demonstrating a series of catastrophic attacks on the app which completely breaks security and anonymity. The authors list privacy, authenticity, and reliability as key features needed for mesh messaging systems in a protest setting, but specify no formal security model. Also, Bridgefy has put out an update \cite{rios_2020} after this work to use Signal’s protocol for encrypting messages, but this does not solve many of the core problems related to metadata that our work addresses. Additionally, further implementation problems in Bridgefy leading to vulnerabilities were found by Eikenberg~\etal~\cite{eikenberg_albrecht_paterson_2021}.

\subsection{Anonymous Communication Systems}

Gathering metadata is a common practice among governments, leading to many systems attempting to hide this. The design and construction of anonymous communication systems has been an active area of research for decades, including the early work of Chaum~\cite{Chaum81,Chaum88}, mixnets \cite{Chaum81, mixminion, tarzan, riffle, atom, xrd, yodel} that present an anonymity/latency tradeoff \cite{trilemma}, DC-nets \cite{Chaum88, herbivore, dissent, dissentnumbers, verdict, riposte}, and MPC techniques \cite{asynchromix, blinder, mcmix}. Shuffling techniques have been improved by \cite{riffle, pynchon, dissentnumbers}, and mixing techniques have been improved to support more users \cite{atom, xrd}. Many anonymity systems have been proposed in the research literature such as Riposte \cite{riposte} and Blinder \cite{blinder} which use DPFs \cite{dpfs}, Talek \cite{talek}, Express \cite{express},  Pung \cite{pung}, and Clarion \cite{clarion}.

The most widely deployed system today is Tor~\cite{tor} which has been analyzed by \cite{BKM+16, KBS+19, KBS20}. Unfortunately, Tor is vulnerable to traffic analysis by passive adversaries with a global view of the network \cite{breaktor, JWJ+13, HB13} and this may be an inherent limitation of the system \cite{trilemma, trilemma2}. HORnet \cite{hornet}, Loopix \cite{loopix}, and TARAnet \cite{taranet} provide stronger resistance to passive attacks. Differential privacy-based approaches  \cite{Dwork06, vuvuzela, alpenhorn, stadium, karaoke} have also been explored to hide metadata. Encore \cite{Aditya14encore} allows mobile social apps to allow users to securely communicate with those in close proximity using existing network services. Proposed systems commonly rely on having access to servers on the Internet and are not suited to the networking restrictions of our problem. The problem of building anonymous mesh networks with strong privacy guarantees has received relatively little attention. Jansen and Beverly worked on providing anonymity in Delay Tolerant Networks \cite{5680442} and Lerner~\etal worked on allowing users to send out broadcast messages during Internet blackouts \cite{2016arXiv161203371L}. Recently, however,  Albrecht~\etal~\cite{breakingbridgefy} drew attention to the critical need for new solutions in this space by pointing out a number of weaknesses in deployed mesh messaging apps that have seen ad-hoc adoption in highly sensitive settings like protests. This work is an attempt to initiate a principled approach to defining and satisfying the security goals of this new setting. 

\subsection{Local Adversaries in Mesh Networks}
Moby \cite{pradeep2022moby} examines the problem of securely communicating during an Internet blackout while mitigating DoS attacks. It provides forward secrecy, has sender-receiver anonymity, delivers 50.73\% of messages when no adversary is present, and delivers 13.96\% of messages during DoS attacks. This is achieved by limiting the threat model to local adversaries who do not have a global view of the network and can only disrupt local sections of the network. Moby uses a bi-modal communication system that uses the Internet while it is available and switches to a mesh network during a blackout. An implementation is available as an extension to the Signal messaging application.

\subsection{Mesh Routing Algorithms}

The mesh messaging system we describe is an instance of a mobile ad-hoc network. There is a large body of research on algorithms for delivering messages from one device in a mobile ad-hoc network to another (\cite{venkat_mohan_kasiviswanath_2011} surveys). Primarily, these algorithms aim to ensure message delivery only to the destination device, which makes them unsuitable for hiding metadata. We use broadcast algorithms to solve this (\cite{RB2015} surveys). The algorithm we use is based on \cite{BHO+1999}, but could be replaced by another broadcast algorithm which ensures delivery to all devices and does not reveal metadata.

% 

\begin{comment}
Most routing algorithms for these networks fit into three categories: 1) Proactive routing protocols; 2) reactive routing protocols; and 3) Hybrid Routing Protocols. They work as follows:
\begin{myitemize}
	\item Proactive routing protocols have each node keep the state of the entire topology and work best for small networks.
	\item Reactive routing protocols work by having nodes broadcast route discovery packets to the network. Paths for packet delivery are then built backwards from the destination to the source.
	\item Hybrid routing algorithms combine the above two approaches.
\end{myitemize}
\cite{venkat_mohan_kasiviswanath_2011} Unfortunately, proactive routing protocols are not suited for large mesh networks like a protest and reactive routing algorithms require explicitly sharing the metadata that we want to hide.

\textbf{Broadcast algorithms} 
\textcolor{red}{TODO: Bruce is expanding this out.}
\end{comment}
\section{Conclusion and Future Work}
Anonymous communication over mesh networks is growing in importance because of oppressive governments. Unfortunately, little work has been done in tackling this problem and to the best of our knowledge, this is the first proposal of a solution along with a security definition that meets the needs of protesters. Our protocol hides the contents of all messages and metadata except for the fact that someone is participating in the protocol. We modeled the protocol and showed that it can support protests of a few hundred to a few thousand users, depending on factors such as crowd density and message frequency. Finally, we presented some alternatives to parts of the protocol's design and outlined how to greatly increase its scaling at the cost of some anonymity or restrictions on how many friends the protocol allows a user to have.

% The following are interesting directions for future work:
We leave the following to future work:
\vspace{-1ex}
\begin{myitemize}
\item Modeling and implementing the clique broadcasting extension to support much larger protests
\item Examining weaker security definitions and what efficiency improvements they allow
\item Creating unicast algorithms that scale to large mesh networks that also sufficiently hide metadata
\item Preventing DoS attacks at the protocol level
\item Measuring power consumption of this protocol and its effects on battery life for various devices
\end{myitemize}
\vspace{-1.5ex}

\section{Acknowledgements}
We would like to thank Amalia Perry, Aidan Perry, Marie Perry, Deepak Kumar, Alex Ozdemir, Sanket Gupte, Gerry Wan, Kimberly Ruth, Hans Hanley, Liz Izhikevich, and Wilson Nguyen for writing feedback, debugging, and support. Google provided GCP credits to evaluate the model. This work was funded by NSF, DARPA, a grant from ONR, and the Simons Foundation. Opinions, findings, and conclusions or recommendations expressed in this material are those of the authors and do not necessarily reflect the views of these organizations.

{\small
\setstretch{0.95}
\bibliographystyle{abbrv}
\bibliography{reference, Broadcast_algorithms}
}

\appendices
\section{Model Implementation Details} \label{sec:sim_details}
% We leave the description of the model implementation to the full version of this paper.
The model takes the following as parameters: \ms{Minutes} (how many minutes of real time should the protest run for in the model, i.e. 5 would mean a protest that happens for 5 minutes in the real world and it may take more or less cpu time to run the model to completion), \ms{Devices} (how many devices participate in the model), \ms{Placement} (how people are placed in the model, i.e. uniformly in a grid or distributed normally), \ms{Spacing} (how far people are apart in the grid if uniform placement was chosen), \ms{Movement} (if people are moving or standing still), \ms{Broadcast\_Type} (simple or smart broadcasting), \ms{DS\_Interval} (how frequently to exchange a \ms{DS} if smart broadcasting was selected), and \ms{Rate} (how frequently to send messages).

The model works as follows. Time is broken up into discrete 100 ms intervals. At setup, every device is assigned a \ms{start\_time} randomly from a 1 minute time window using these intervals, i.e. \ms{60 * 1000 / 100 = 600} possible starting times. These devices are also placed on a 2D grid according to the provided \ms{Placement} parameter and they are connected with their neighbors. Then, the model proceeds by running a sequence of steps for each time interval until the value specified by \ms{Minutes} is reached. At each time interval, each device checks if they should send a message. If they should, a random message is generated and is added to the device's \ms{DS}. The sender's outgoing bandwidth for this time interval is updated and various values are recorded for bookkeeping and log processing. The message is sent to each neighbor and their incoming bandwidths are updated. If smart broadcasting is being used, then each device checks if they should send their \ms{DS} to neighbors. If they should, their outgoing bandwidth is updated, each neighbor receives a copy of the \ms{DS}, updates their list of missing messages, and their incoming bandwidth is updated. Then each device broadcasts a request which consists of a list of messages that they are missing. Each neighbor's incoming bandwidth is recorded and their list of requests to respond to is updated. Then, each device checks their list of requests to respond to, takes the union of them, and broadcasts this to all neighbors. The sender's outgoing bandwidth is updated, each neighbor's incoming bandwidth is updated, and each neighbor adds all messages that they receive that they did not already have in their \ms{DS}. Finally, each device removes all entries from their \ms{DS} that have expired. If simple broadcasting was used instead, devices keep a queue of received messages. Each device sends each message in the queue to all neighbors. Each neighbor adds it to their queue to be processed in the next round if they have not already seen it. The sender's outgoing bandwidth and the receivers' incoming bandwidths are recorded as required. Finally, if \ms{Movement} was set to have devices move, each device moves and their list of neighbors is updated. This loop continues until the desired time in the model is reached.

Every time that a new message, a request for missing messages, a response to missing messages, or a \ms{DS} is transmitted, it is logged to a file. These are the 4 types of communications. Specifically, the unique id of the sender, the unique id of the receiver, the communication type, a unique id of the communication, a unique transmission id to group together broadcasts, the size in bytes of the communication, a timestamp of the communication, and a unique id of the model that it was a part of are logged. Various parameters of the model such as the length of time in the model, the number of devices, the rate at which a \ms{DS} was exchanged, if devices were moving, how devices were distributed, and the broadcasting algorithm are logged at the start of the model to a separate log along with a unique id to join the two types of logs together. These logs are later processed with Pandas for analysis and to build the graphs presented in section \ref{sec:eval}.

\section{Problems with Unicast Alternatives}
\label{unicast_attacks}
A natural idea is to set up ``broadcast bubbles,'' where a message is broadcast to all nodes within a fixed number of hops, on both ends of the unicast path. The sender and receiver each choose a nearby node as a proxy. Then, both broadcast their messages within a fixed range, their proxies pick them up, send them to the other proxy, and then that proxy broadcasts the message within a fixed range. The key concept is that the only thing announced to the network is that the two proxies are communicating, effectively giving nodes an anonymity set of all nodes in their proxy’s broadcast bubble. Unfortunately, this leads to triangulation attacks, where a group of malicious nodes in the broadcast bubble can use timing information to learn about the sender (Figure \ref{fig:triangulation_attack}).

If broadcast bubbles are used, an attacker can learn who sent a message by controlling a few nodes in the anonymity set. If the broadcast messages contain a ttl (broadcast is within a fixed range), the adversary can look at the topology and know that the target is original$\_$ttl - current$\_$ttl hops away. This gives a set of candidates for each node they control. The intersection of this set can, in the worst case, be one node. This reveals the identity of who is using the proxy. A similar attack can be launched by looking at the timestamps of when each node received the message and comparing them.

Another idea that unfortunately does not work is having Bluetooth devices periodically change their addresses. While this may sufficiently hide routing information allowing for more traditional and more efficient routing algorithms for mesh networks to be used, there are cases where users' previous ids can always be linked with their new ids. One such case is when a device only has one neighbor, i.e. someone standing at the very edge of the crowd. This person would see one neighbor with a particular Bluetooth address disappear and then see a new neighbor with a different Bluetooth address appear. These two neighbors are most likely the same person. An adversary could target specific individuals and keep a record of their address changes throughout the protest. They could combine this with the routing information from the protest and deanonymize individuals.

\begin{figure}
\centering
\includegraphics[width=0.4\linewidth]{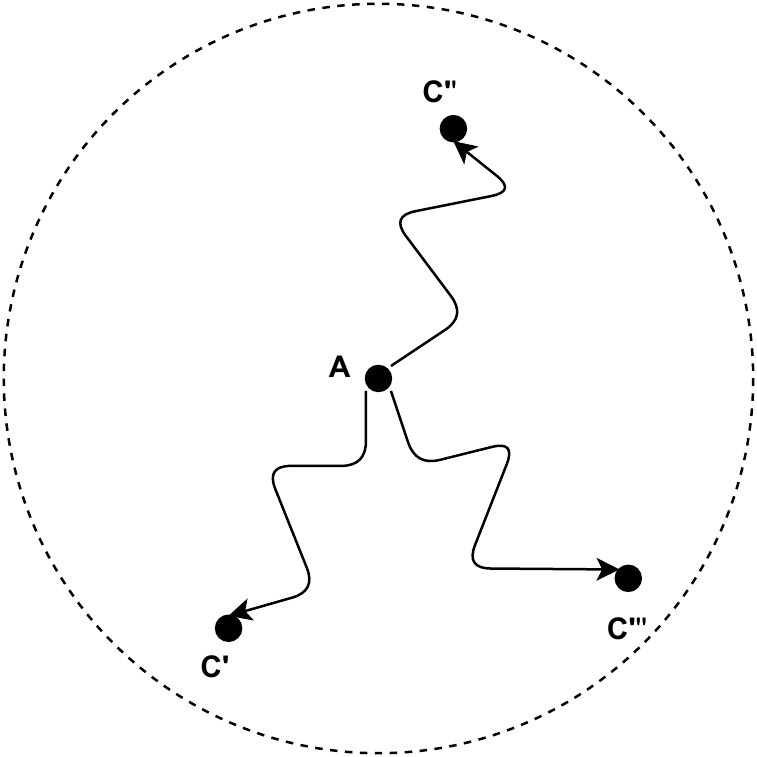}
\caption{The adversary can use timing information from $C^{\prime}$, $C^{\prime\prime}$, and $C^{\prime\prime\prime}$ to learn that $A$ is the device that is broadcasting.}
\label{fig:triangulation_attack}
\vspace{-2.5em}
\end{figure}

\vspace{-0.5em}
\section{Messaging Phase (Abstract)} \label{sec:active_protesting_appendix}
% We leave the Messaging phase using the abstract broadcasting algorithm to the full version of this paper.

Here we describe the Messaging phase in more detail. Although we make use of a preparation phase to establish relationships between users, the Messaging phase algorithms we introduce are equally applicable to a system that establishes relationships in another way. This phase consists of a control loop that repeats until the phase ends. We also make use of an abstract \ms{Broadcast} algorithm that delivers messages to all users in the system. For security, we require that the broadcast algorithm attempts to deliver messages to all users in the system and does not address messages to any particular users. Pseudocode showing this control loop and related functions is presented in Figure \ref{fig:algorithm_simple}. \ms{PK_s} is the device's public key, \ms{SK_s} is the device's secret key, \ms{PK_r} is the intended receiver's public key, and \ms{PK_{dummy}} is a dummy public key. Each device is independently assigned a uniformly random value indicating when the device is allowed to send a message called \ms{turn}. Devices send messages every \ms{send\_rate} ms from \ms{turn}.

\begin{figure}
	\fbox{
		\begin{minipage}[t]{0.5\textwidth-7\fboxsep-7\fboxrule}
			\input{code/join_simple}
		\end{minipage}
	}
	\fbox{
		\begin{minipage}[t]{0.5\textwidth-7\fboxsep-7\fboxrule}
			\input{code/send_simple}
		\end{minipage}
	}
	\fbox{
		\begin{minipage}[t]{0.5\textwidth-7\fboxsep-7\fboxrule}
			\input{code/receive_simple}
		\end{minipage}
	}
	\fbox{
		\begin{minipage}[t]{0.5\textwidth-7\fboxsep-7\fboxrule}
			\input{code/main_simple}
		\end{minipage}
	}
\caption{Main control loop and related functions.}
\label{fig:algorithm_simple}
\vspace{-1em}
\end{figure}

This control loop consists of checking a series of conditions. \ms{HAS\_SEND} and \ms{MESSAGE\_ARRIVED} are boolean values that are \ms{True} when a device has a message to send and a new message has been received respectively and are \ms{False} otherwise.

First, a device checks if it should send a message by seeing if a multiple of \ms{send\_rate} has passed since \ms{turn}. If it is time to send a message, the device then checks \ms{HAS\_SEND} to see if it should send a real message or a dummy message.

Second, when a message is received, the function \ms{receive} is called. This function tries to decrypt the message. If the decryption succeeds, the function \ms{process} that takes as input a message and displays the message to the user is called.

Third, the device forgets that it has seen all messages that have expired. This is done by removing any that are older than \ms{time\_to\_keep}, which is a constant value indicating the amount of time that a device remembers messages. This is done to prevent the device from having to store data on irrelevant messages that have already been received by all devices.

This completes the description of the control loop. When a real or dummy message is sent, \ms{send} is called and it encrypts and broadcasts the message. In order to hide the identity of the receiver, ensure message integrity, and ensure that users cannot impersonate one another, a secure signcryption must be used. The function \ms{join} is called when a device joins the network and is responsible for initializing \ms{turn}. Note that none of these variables need to be synchonized across devices. In fact, having \ms{turn} synchronized across devices would lead to the worst case traffic patterns (see Section \ref{sec:how_much_bandwidth}).

\section{Signcryption Instantiation}
\label{sec:signcryption}
For completeness, we provide an instantiation of the encryption and decryption functions used for signcryption (Figure \ref{fig:signcryption}). The device prepends its PK to the message and signs it. This allows the receiver to know who sent the message and verify their identity. Then the device encrypts the concatenation of the signature and message by generating a symmetric key, encrypting the symmetric key with ElGamal so that the receiver is not leaked, encrypting the signature and message via AES-GCM with the symmetric key, and placing the public key encryption of the symmetric key in the Associated Data field of AES-GCM. Any signcryption that does not leak the sender or receiver can be used and our example is not the only correct choice.
\begin{figure}
	\fbox{
		\begin{minipage}[t]{0.45\textwidth}
			\input{code/encrypt} 
		\end{minipage}
	}
	\fbox{
		\begin{minipage}[t]{0.45\textwidth}
			\input{code/decrypt} 
		\end{minipage}
	}
\caption{Signcryption Functions.}
\label{fig:signcryption}
\vspace{-2em}
\end{figure}

% \section{Simulation Proof}
% \label{sec:sim_proof}
% \input{simulation_proof}

\section{Public Key Anonymity}
\label{sec:public_key_anonymity}
Bellare~\etal defined Key-Privacy of public key ciphers \cite{anonymity_original}. Key-Privacy (Anonymity) is the property that an adversary who is given a ciphertext $c$ of some message $m$ cannot determine which public key from a known set was used to encrypt $c$. Bellare~\etal's definition uses a one time game, where the adversary submits a pair of public keys along with a message $(pk_0, pk_1, m)$ to a challenger which returns $E(pk_b, m)$ for $b \in\zo$. The adversary then attempts to output the value $b$. Our public key anonymity property extends this and follows from Bellare~\etal's definition by allowing the adversary to submit multiple queries $(pk_0, pk_1, m)_i$ for $i < Q$, where $Q \in \mathbb{Z}$ and $Q < poly(\lambda)$. The adversary then attempts to output the value $b$.

\section{Security Definitions and Proofs}
\label{sec:security_games}
\begin{definition}[Mesh Confidentiality]
For a given protocol $\Pi$, adversary $\calA$, security parameter $\lambda$, integers $N,T,h\in\Z$, and bit $b\in\zo$, we define the mesh confidentiality experiment $\ms{MCONF}[\Pi, \calA, \lambda, N, T, h, b]$. The experiment consists of a challenger $\calC$ who is given $b$ as an input and an adversary $\calA$. It proceeds as follows.

\noindent\textbf{Setup.} In the setup phase, the adversary $\calA$ sends $M$ to the challenger $\calC$, where $M \subset \{1,...,N\}$, $|M| = N - h$ represents the set of malicious users, as well as a set of public keys $pk_i$ for $i \in M$. In response, $\calC$ generates public keys $pk_i$ for honest users $i \in \{1,...,N\} \setminus M $ and sends them to $\calA$.

\noindent\textbf{Messaging.} The messaging phase continues for $T$ rounds. All plaintext messages have a fixed length $\ell$ In each round, the adversary sends the following to the challenger.
\vspace{-1ex}
\begin{myitemize}
    \item $G_i$ which is a graph on N nodes
    \item $\{(\ms{sender}, \ms{rec}_0, m_0, \ms{rec}_1, m_1)_i\}_{i \in [Q]}$, where $Q \in\Z$ is the number of queries sent this round, subject to the restriction that if $\ms{rec}_0\in M$ or $\ms{rec}_1\in M$, then $(\ms{rec}_0, m_0) = (\ms{rec}_1, m_1)$ $\forall i$. If this restriction is violated, the challenger aborts and the experiment outputs 0. Note that $\ms{rec}_0$ and $\ms{rec}_1$ are the public keys of the receivers and not the nodes in $G_i$.
    \item $\{(\ms{pkt}_i, \ms{rec}_i)\}_{i\in M}$, the set of packets (encrypted messages) sent by malicious users and the nodes (not public keys) in $G_i$ to which they are sent. Multiple packets can be sent per round for $i \in M$.
\end{myitemize}
\vspace{-1ex}

In response, the challenger processes $(\ms{sender}, \ms{rec}_b, m_b)$ from the adversary as well as $(m_i, \ms{rec}_i)$ for all $i \in M$. The challenger internally runs all honest parties and responds with all messages from honest users to malicious users that can be delivered given network topology $G_i$ along with all ciphertexts passing through any node corresponding to $i \in M$.

\noindent\textbf{Output.} In the output phase, the adversary outputs a bit $b'\in\zo$, which is also the output of the experiment. 

\noindent For $b = 0,1$, let $W_b$ be the event that $\calA$ outputs 1 in Experiment $b$. We define $\calA$'s advantage with respect to the protocol $\Pi$, security parameter $\lambda$, and integers $N,T,h$ as 
\[\ms{CONFadv}(\calA, \Pi, \lambda, N, T, h) = \Big|Pr[W_0] - Pr[W_1]\Big|.\]
\noindent We say $\Pi$ has \emph{mesh confidentiality} if for all PPT adversaries~$\calA$, and for all $\lambda, N, T, h$, 
\[\ms{CONFadv}(\calA, \Pi, \lambda, N, T, h)  \leq \ms{negl}(\lambda).\]
\end{definition}

\begin{definition}[Mesh Integrity]
For a given protocol $\Pi$, adversary $\calA$, security parameter $\lambda$, and integers $N,T,h\in\Z$, we define the mesh integrity experiment $\ms{MINT}[\Pi, \calA, \lambda, N, T, h]$. The experiment consists of a challenger $\calC$ and an adversary $\calA$. It proceeds as follows.

\noindent\textbf{Setup.} In the setup phase, the adversary $\calA$ sends $M$ to the challenger $\calC$, where $M \subset \{1,...,N\}$, $|M| = N - h$, represents the set of malicious users, as well as a set of public keys $pk_i$ for $i \in M$. In response, $\calC$ generates public keys $pk_i$ for honest users $i \in \{1,...,N\} \setminus M $ and sends them to $\calA$. $\calC$ initializes sets $S$ and $R$ that each hold tuples of the form $(sender, rec, m, ct)$, where $ct$ is the ciphertext corresponding to $(sender, rec, m)$, to keep track of what was sent and what was accepted by honest users respectively.

\noindent\textbf{Messaging.} The messaging phase continues for $T$ rounds. All plaintext messages have a fixed length $\ell$ In each round, the adversary sends the following to the challenger.
\vspace{-1ex}
\begin{myitemize}
    \item $G_i$ which is a graph on N nodes
    \item $\{(\ms{sender}, \ms{rec}, m)_i\}_{i \in [Q]}$, where $Q \in\Z$ is the number of queries sent this round and $\ms{sender}$ is honest. If this condition is violated, the challenger aborts. Note that $\ms{sender}$ and $\ms{rec}$ are the public keys of the sender and receiver, not the nodes in $G_i$.
    \item $\{(\ms{pkt}_i, \ms{rec}_i)\}_{i\in M}$, the set of packets (encrypted messages) sent by malicious users and the nodes (not public keys) in $G_i$ to which they are sent. Multiple packets can be sent per round for $i \in M$.
\end{myitemize}
\vspace{-1ex}

In response, the challenger processes $\{(\ms{sender}, \ms{rec}, m)_i\}_{i \in [Q]}$ from the adversary, producing a $ct$ for each, and adds $(\ms{pk}_\ms{sender}, \ms{rec}, m, ct)$ for each one to $S$ and also processes all of the messages from the malicious users. The challenger internally runs all honest parties and responds with all messages from honest users to malicious users that can be delivered given network topology $G_i$ along with all ciphertexts passing through any node corresponding to $i \in M$. Whenever an honest user $i$ receives a message $ct$, it attempts to decrypt it. If the output is not $\bot$ (i.e., it is a $(\ms{pk}_\ms{sender}, m)$ pair), then $(\ms{pk}_\ms{sender}, i, m, ct)$ is added to $R$ if $\ms{sender}$ is an honest user.

\noindent\textbf{Output.} In the output phase, the experiment outputs $0$ if $R\subseteq S$, $1$ otherwise.

\noindent For $b = 0,1$, let $W_b$ be the event that the challenger outputs 1 after interacting with $\calA$. We define $\calA$'s advantage with respect to the protocol $\Pi$, security parameter $\lambda$, and integers $N,T,h$ as 
\begin{multline*}
\ms{MINTadv}(\calA, \Pi, \lambda, N, T, h) \\
= Pr[\ms{MINT}[\Pi, \calA, \lambda, N, T, h] = 1]
\end{multline*}
\noindent We say $\Pi$ has \emph{mesh integrity} if for all PPT adversaries~$\calA$, and for all $\lambda, N, T, h$,
\[\ms{MINTadv}(\calA, \Pi, \lambda, N, T, h)  \leq \ms{negl}(\lambda).\]

\end{definition}

\begin{theorem}[\ref{thm:int}]
Assuming signcryption scheme $(G,E,D)$ has ciphertext integrity, $\Pi$ has mesh integrity.
\end{theorem}

\begin{proof}
The proof of Theorem~\ref{thm:int} proceeds by a check and guess argument. Throughout the proof, let $H$ be the set of honest users, i.e., $H = \{1,...,N\} \setminus M$. WLOG, we renumber the users so that the first $h$ are honest.

 Assume adversary $\calA$ exists that has a non-negligible advantage in $\ms{MINT}[\Pi, \calA, \lambda, N, T, h]$. We will construct an adversary $\calB$ that uses $\calA$ to break the ciphertext integrity of $(G,E,D)$. Let $\ms{Chal}_\ms{int}$ be a challenger to the signcryption ciphertext integrity game~\cite{boneh_shoup}. $\calB$ acts as the adversary to $\ms{Chal}_\ms{int}$ and the challenger to $\calA$.
 
 $\calB$ works as follows corresponding to the mesh integrity game's phases:
 
\noindent Setup Phase:
\vspace{-1ex}
\begin{myitemize}
    \item $\calB$ chooses a uniformly random $(i, j) \in H^2$
    \item $\calB$ sends $(i, j)$ to $\ms{Chal}_\ms{int}$ which generates $((sk_i, pk_i),(sk_j, pk_j))$ and sends back $(pk_i, pk_j)$
    \item When $\calA$ sends $M$ and $pk_i$ for $i \in M$, $\calB$ generates $pk_i$ for $i \in H \setminus \{i, j\}$
    \item $\calB$ sends $pk_i$ for $i \in H$ to $\calA$
\end{myitemize}
\vspace{-1ex}

\noindent Messaging Phase:
\vspace{-1ex}
\begin{myitemize}
    \item When When $\calA$ sends $G_i$, $\calB$ stores it for processing messages in the current round.
    \item When $\calA$ sends $(\ms{sender}, \ms{rec}, m)$, $\calB$ encrypts it and adds it along with its corresponding $ct$ to $S$. If $(\ms{sender}, \ms{rec}) = (i, j)$, $\calB$ makes an $S \rightarrow R$ encryption query to $\ms{Chal}_\ms{int}$ which responds with the corresponding ciphertext $ct$. If $(\ms{sender}, \ms{rec}) \neq (i, j)$ but $sender \in \{i, j\}$, $\calB$ makes an $X \rightarrow Y$ encryption query to $\ms{Chal}_\ms{int}$ which responds with the corresponding ciphertext $ct$. $ct$ is used while processing $A$'s query and the honest nodes.
    \item When $\calB$ delivers messages $ct$ from $(\ms{sender}, \ms{rec}, m)$ queries, to the node indicated by $\ms{rec}$, $\calB$ adds $(\ms{sender}, \ms{rec}, m, ct)$ to $R$ if it decrypts correctly. $\calB$ checks this by sending $ct$ to $\ms{Chal}_\ms{int}$ in an $X \rightarrow Y$ decryption query and checks if the response is $\bot$.
    \item When $\calA$ sends $(\ms{pkt}, \ms{rec})$ and $\calB$ is processing each honest node receiving it, if the honest node is $i$ or $j$, $\calB$ makes 2 $X \rightarrow Y$ decryption queries to $\ms{Chal}_\ms{int}$. The first where $i$ is the sender and $j$ is the receiver and the second with $i$ and $j$'s roles reversed. If either plaintext responses denoted by $m$ from $\ms{Chal}_\ms{int}$ are not $\bot$, then this is added to $R$ for $i$ and $j$ in their respective roles of the successful $X \rightarrow Y$ decryption query. In other words, $\calB$ adds $(\ms{sender}, k, m, \ms{pkt})$ to $R$, where $k \in \{i,j\}$ is the honest node being processed.
\end{myitemize}
\vspace{-1ex}

\noindent Output Phase:
\vspace{-1ex}
\begin{myitemize}
    \item $\calB$ outputs $0$ if $R\subseteq S$, $1$ otherwise in the role of $\calA$'s challenger.
    \item If $R\not\subseteq S$ ($\calA$ has successfully forged a message), $\calB$ selects a random entry of $R$ that is not in $S$. If $sender = i$ and $rec = j$, $\calB$ sends the corresponding $ct$ to $\ms{Chal}_\ms{int}$ as the candidate forgery. If $sender \neq i$ or $rec \neq j$, $\calB$ fails in the ciphertext integrity game.
\end{myitemize}
\vspace{-1ex}

When $\calA$ successfully wins the mesh integrity game, $\calB$ guesses the corresponding $i$ and $j$ to $\calA$'s forgery with probability $1/h^2$. Therefore,
\small
\begin{align*}
\noindent&\frac{1}{h^2} \ms{MINTadv}(\calA, \Pi, \lambda, N, T, h) \leq \ms{SINTadv}(\calB, (G,E,D)),
% &\frac{1}{h^2} \ms{MINTadv}(\calA, \Pi, \lambda, N, T, h) 
% &\ \;\leq \ms{SINTadv}(\calB, (G,E,D)),
\end{align*}
\normalsize
so we have that 
\begin{multline*}
    \ms{MINTadv}(\calA, \Pi, \lambda, N, T, h) \\
    \leq h^2 \ms{SINTadv}(\calB, (G,E,D)) < negl(\lambda)
\end{multline*}

\noindent$\ms{MINTadv}(\calA, \Pi, \lambda, N, T, h)$ must be negligible, which is a contradiction. Therefore, $\Pi$ has mesh integrity.
\end{proof}

\begin{theorem}[\ref{thm:conf}]
Assuming signcryption scheme $(G,E,D)$ is a CPA secure cipher and provides public key anonymity, then $\Pi$ has mesh confidentiality.
\end{theorem}

\begin{proof}
%\noindent Let $H = N - |M|$ honest users be assigned $i \in [1,..., H]$
The proof of Theorem~\ref{thm:conf} proceeds by a series of hybrids. Throughout the proof, let $H$ be the set of honest users, i.e., $H = \{1,...,N\} \setminus M$. 

\noindent$\ms{Hyb}_0$: This is the mesh confidentiality experiment $\ms{MCONF}[\calA, \Pi, \lambda, N, T, h, 0]$. 

\noindent$\ms{Hyb}_1$: We modify the behavior of the challenger such that when it receives a message of the form $(\ms{sender}, \ms{rec}_0, m_0, \ms{rec}_1, m_1)$, if $\ms{sender}\in H$, $\ms{rec}_0\in H$, and $\ms{rec}_1\in H$, it replaces the message with $(\ms{sender}, \ms{rec}_0, m_0'=0^\ell, \ms{rec}_1, m_1' = 0^\ell)$ before processing it. Otherwise, $m_0'=m_0$ and $m_1'=m_1$. So, plaintext messages sent between honest users are replaced with encryptions of 0, regardless of whether their ciphertexts are given to malicious users along the way. 

% We prove that this hybrid is indistinguishable from the preceding one by the CPA security of cipher $(E,D)$ in the full version of this paper. 

In Lemma~\ref{lemma:hyb01}, we prove that this hybrid is indistinguishable from the preceding one by the CPA security of cipher $(E,D)$. 

\noindent$\ms{Hyb}_2$: We modify the behavior of the challenger such that when it receives a message of the form $(\ms{sender}, \ms{rec}_0, m_0, \ms{rec}_1, m_1)$, after it computes $m_0',m_1'$, it processes the message $(\ms{sender}, m_1', \ms{rec}_1)$. That is, it switches the message processed from being the message sent to $\ms{rec}_0$ to being the message sent to $\ms{rec}_1$. 

% We prove that this hybrid is indistinguishable from the preceding one by the anonymity of cipher $(E,D)$ in the full version of this paper. 

In Lemma~\ref{lemma:hyb12}, we prove that this hybrid is indistinguishable from the preceding one by the anonymity of cipher $(E,D)$. 

\noindent$\ms{Hyb}_3$: We reverse the changes made in $\ms{Hyb}_1$ by not changing messages from the adversary to $(\ms{sender}, \ms{rec}_0, m_0', \ms{rec}_1, m_1')$ before processing them. That is, the messages sent between honest clients are the actual messages requested by the adversary. 

% We prove that this hybrid is indistinguishable from the preceding one by the CPA security of cipher $(E,D)$ in the full version of this paper. Note that this hybrid is identical to the mesh confidentiality experiment $\ms{MCONF}[\calA, \Pi, \lambda, N, T, h, 1]$. 

In Lemma~\ref{lemma:hyb23}, we prove that this hybrid is indistinguishable from the preceding one by the CPA security of cipher $(E,D)$. Note that this hybrid is identical to the mesh confidentiality experiment $\ms{MCONF}[\calA, \Pi, \lambda, N, T, h, 1]$. 

% We now show that each of the preceding hybrids are computationally indistinguishable. The proof of the theorem follows from the triangle inequality and the lemmas in the full version of this paper.\end{proof}

We now show that each of the preceding hybrids are computationally indistinguishable. The proof of the theorem follows from the triangle inequality and the lemmas below. 

\end{proof}

\begin{lemma}\label{lemma:hyb01}
Assuming that $(E,D)$ is a CPA-secure cipher, $\ms{Hyb}_0$ is computationally indistinguishable from $\ms{Hyb_1}$. 
\end{lemma}
\begin{proof}
To prove the lemma, we proceed via a sequence of inner hybrid experiments and WLOG, renumber the users so that the first $h$ are honest. For $\gamma \in [H]$, we define the hybrid experiments $\ms{Hyb}_{0, \gamma}$ as follows:
\begin{myitemize}
    \item The challenger proceeds through the setup phase identically to the mesh confidentiality experiment $\ms{MCONF}[\calA, \Pi, \lambda, N, T, h, 0]$.
    \item The challenger proceeds through the messaging phase identically to $\ms{MCONF}[\calA, \Pi, \lambda, N, T, h, 0]$, except when processing $\calA$'s queries, if $sender \in H $ and $rec_j \leq \gamma$, where $j \in\zo$, the query $(\ms{sender}, \ms{rec}_0, m_0, \ms{rec}_1, m_1)$ has $m_j$ replaced with $0^\ell$.
    \item The challenger proceeds through the output phase identically to $\ms{MCONF}[\calA, \Pi, \lambda, N, T, h, 0]$.
\end{myitemize}

\noindent By definition, experiment $\ms{Hyb}_{0, 0}$ corresponds to experiment $\ms{MCONF}[\calA, \Pi, \lambda, N, T, h, 0]$ and experiment $\ms{Hyb}_{0, h}$ corresponds to experiment $\ms{Hyb}_{1}$.

\noindent We now show that each consecutive hybrid experiments $\ms{Hyb}_{0, \gamma - 1}$ and $\ms{Hyb}_{0, \gamma}$, for $\gamma \in [H]$, are computationally indistinguishable.

\noindent Assume adversary $\calA$ exists that can distinguish between $\ms{Hyb}_{0, \gamma - 1}$ and $\ms{Hyb}_{0, \gamma}$. We will construct an adversary $\calB$ that uses $\calA$ to break CPA security. Let $Chal_{cpa}$ be a challenger to the CPA security game. $\calB$ acts as the adversary to $Chal_{cpa}$ and the challenger to $\calA$.

\noindent $\calB$ receives $pk_{\gamma}$ from $Chal_{cpa}$.

\noindent $\calB$ proceeds through the setup phase identically to $\ms{Hyb}_{0, \gamma - 1}$, except when generating the set of public keys, $\calB$ uses $pk_{\gamma}$ as $\gamma$'s public key.

\noindent $\calB$ proceeds through the messaging phase identically to $\ms{Hyb}_{0, \gamma - 1}$, except whenever it needs to encrypt a message $m$ to receiving user $\gamma$, it sends $(\ms{msg}_0 = m, \ms{msg}_1 = 0^\ell)$ to $Chal_{cpa}$. $Chal_{cpa}$ computes $c \leftarrow E(pk, \ms{msg}_b)$ and sends $c$ to $\calB$. $\calB$ uses $c$ for the message to $\gamma$. Whenever $\calB$ needs to decrypt a message for $\gamma$, $\calB$ it does not attempt decryption. Note that this step is hidden from $\calA$ and does not affect the messages it receives from the challenger.

\noindent In the output phase, $\calB$ passes along $\calA$'s output to $Chal_{cpa}$ as its own output.

\noindent $\calB$ provides a perfect simulation of $\ms{Hyb}_{0, \gamma - 1}$ if it uses $\calA$'s messages in the queries from $\calA$ (this is the case when $b=0$) and a perfect simulation of $\ms{Hyb}_{0, \gamma}$ if it uses $0^\ell$ in the queries from $\calA$ (this is the case when $b=1$). Therefore, with the same distinguishing advantage of the two experiments by $\calA$, $\calB$ breaks the CPA security of $(E, D)$. Since $(E, D)$ is a CPA secure cipher, this is a contradiction.

\noindent Since $h$ is a fixed constant less than or equal to $N$, the indistinguishability of the subsequent inner hybrids proves Lemma \ref{lemma:hyb01}.

\end{proof}

\begin{lemma}\label{lemma:hyb12}
Assuming that cipher $(E,D)$ provides anonymity, $\ms{Hyb}_1$ is computationally indistinguishable from $\ms{Hyb_2}$. 
\end{lemma}
\begin{proof}
Assume adversary $\calA$ exists that can distinguish between $\ms{Hyb}_{1}$ and $\ms{Hyb}_{2}$. We will construct an adversary $\calB$ that uses $\calA$ to break public key anonymity. Let $\ms{Chal}_\ms{anon}$ be a challenger to the anonymity security game. $\calB$ acts as the adversary to $\ms{Chal}_\ms{anon}$ and the challenger to $\calA$.

\noindent $\calB$ proceeds through the setup phase identically to $\ms{Hyb}_{1}$ but receives $pk_i \forall i \in H$.

\noindent $\calB$ proceeds through the messaging phase identically to $\ms{Hyb}_{1}$, except when it receives a query $(\ms{sender}, \ms{rec}_0, m_0, \ms{rec}_1, m_1)$ where $sender \in H$, $rec_0 \in H$, and $rec_1 \in H$. When this occurs, $\calB$ sends $(pk_{rec_0}, pk_{rec_1}, 0^\ell)$ to $\ms{Chal}_\ms{anon}$. All messages satisfying the modification criteria were replaced with $0^\ell$ in $\ms{Hyb}_{1}$, so this is the correct message to use in queries to $\ms{Chal}_\ms{anon}$. $\ms{Chal}_\ms{anon}$ computes $c \leftarrow E(pk_b, m)$ and sends $c$ to $\calB$. $\calB$ uses $c$ for the message represented by $\calA$'s query. Since all honest nodes are sending messages at regular intervals and all ciphertexts will be delivered to all nodes conditioned on $G_i$, changing the recipient of a message does not change which nodes receive ciphertexts in any round.

\noindent In the output phase, $\calB$ passes along $\calA$'s output to $\ms{Chal}_\ms{anon}$ as its own output.

\noindent $\calB$ provides a perfect simulation of $\ms{Hyb}_{1}$ if it uses $rec_0$ in the queries from $\calA$ (this is the case when $b=0$) and a perfect simulation of $\ms{Hyb}_{2}$ if it uses $rec_1$ in the queries from $\calA$ (this is the case when $b=1$). Therefore, with the same distinguishing advantage of the two experiments by $\calA$, $\calB$ breaks the Anonymity security of $(E, D)$.

\noindent Since $(E, D)$ is an anonymous cipher, this is a contradiction and proves Lemma \ref{lemma:hyb12}.

\end{proof}

\begin{lemma}\label{lemma:hyb23}
Assuming that $(E,D)$ is a CPA-secure cipher, $\ms{Hyb}_2$ is computationally indistinguishable from $\ms{Hyb_2}$. 
\end{lemma}
\begin{proof}
The proof is identical to the proof of Lemma~\ref{lemma:hyb01}. 
\end{proof}

\end{document}